\documentclass{conm-p-l}

\usepackage[french, english]{babel}

\usepackage{mathrsfs}  

\usepackage{xcolor}
\usepackage{graphicx}

\usepackage{upgreek}

\usepackage{hyperref}

\usepackage{amsmath}
\usepackage{amssymb}

\usepackage{multicol}
\usepackage{enumerate}

\newtheorem{theorem}{Theorem}[section]

\theoremstyle{definition}
\newtheorem{definition}[theorem]{Definition}
\newtheorem{claim}[theorem]{Claim}
\newtheorem{consequence}[theorem]{Consequence}
\newtheorem{proposition}[theorem]{Proposition}
\newtheorem{example}[theorem]{Example}

\theoremstyle{remark}
\newtheorem{remark}[theorem]{Remark}

\numberwithin{equation}{section}

\newcommand{\abs}[1]{\lvert#1\rvert}

\DeclareSymbolFont{UPM}{U}{eur}{m}{n}
\DeclareMathSymbol{\uppartial}{0}{UPM}{"40}

\newcommand{\slfrac}[2]{\left.#1\middle/#2\right.}

\newcommand{\ensemblenombre}[1]{\mathbb{#1}}

\newcommand{\Z}{\ensemblenombre{Z}}
\newcommand{\Q}{\ensemblenombre{Q}}
\newcommand{\R}{\ensemblenombre{R}}

\newcommand{\opd}{\mathop{}\mathopen{}\mathrm{d}}
\newcommand{\opdelta}{\mathop{}\mathopen{}\updelta}
\newcommand{\oppartial}{\mathop{}\mathopen{}\uppartial}

\newcommand{\trace}[1]{\mathrm{Tr}\mathopen{}\left(#1\right)}

\newcommand{\Hom}[2]{\mathrm{Hom}\mathopen{}\left(#1,#2\right)}
\newcommand{\Ext}[2]{\mathrm{Ext}\mathopen{}\left(#1,#2\right)}

\newcommand{\noyau}{\mathop{}\mathopen{}\mathrm{Ker\,}}
\newcommand{\image}{\mathop{}\mathopen{}\mathrm{Im}\,}

\begin{document}

\title{Generalized Abelian Turaev-Viro and $\mathrm{U}\!\left(1\right)$ BF Theories}

\author{Emil H{\o}ssjer}
\address{Faculty of Science, Department of Mathematics and Computer Science, Centre for Quantum Mathematics, Campusvej 55, 5230 Odense, Denmark}
% \curraddr{Department of Mathematics and Statistics, Case Western Reserve University, Cleveland, Ohio 43403}
\email{hossjer@imada.sdu.dk}
% \thanks{Support information for the second author.}
\author{Philippe Mathieu}
\address{Institut f\"ur Mathematik, Universit\"at Z\"urich, Winterthurerstrasse 190, CH-8057 Z\"urich}
% \curraddr{Department of Mathematics and Statistics, Case Western Reserve University, Cleveland, Ohio 43403}
\email{philippe.mathieu@math.uzh.ch}
\thanks{Ph. M. was supported by the NSF grant 1947155 and the JTF grant 6152, then by the NSF Grant 200020-192080 of the Simons Collaboration on Global Categorical Symmetries, the COST Action 21109 - Cartan geometry, Lie, Integrable Systems, quantum group Theories for Applications (CaLISTA), and the NCCR SwissMAP, funded by the Swiss National Science Foundation. 
}
%
%    Information for third author
\author{Frank Thuillier}
\address{Laboratoire d'Annecy de Physique Th\'{e}orique (LAPTh), 9 Chemin de Bellevue, 74940 Annecy, France}
% \curraddr{Department of Mathematics and Statistics, Case Western Reserve University, Cleveland, Ohio 43403}
\email{frank.thuillier@lapth.cnrs.fr}
% \thanks{Support information for the second author.}
%
%    General info
\subjclass[2020]{Primary 57K16; Secondary 81T25, 81T27}
\date{January 1, 1994 and, in revised form, June 22, 1994.}

% 57K16 Finite-type and quantum invariants, topological quantum field theory (TQFT)
% 81T25 Quantum field theory on lattices
% 81T27 Continuum limits in quantum field theory

\dedicatory{This paper is dedicated to our colleague and friend {\'E}ric Pilon.}

\keywords{Differential cohomology, quantum invariants of manifolds, quantum field theory, gauge theory, $\mathrm{U}\!\left(1\right)$ BF theory}

\begin{abstract}
We explain how it is possible to study $\mathrm{U}\!\left(1\right)$ BF theory over a connected closed oriented smooth $3$-manifold in the formalism of path integral thanks to Deligne-Beilinson cohomology. We show how we can straightforwardly extend the definition to families of theories in any dimension. We extend then the definition of the Turaev-Viro invariant of a connected closed oriented smooth $3$-manifold in an Abelian framework to a family of invariants in any dimension. We show that those invariants can be written as discrete BF theories. We explain how the extensions of $\mathrm{U}\!\left(1\right)$ BF theory we defined can be related to the extensions of Turaev-Viro invariant we constructed.
\end{abstract}

\maketitle

\tableofcontents

\section{Introduction}

\subsection{A bit of history}

\begin{figure}
\begin{center}
\includegraphics{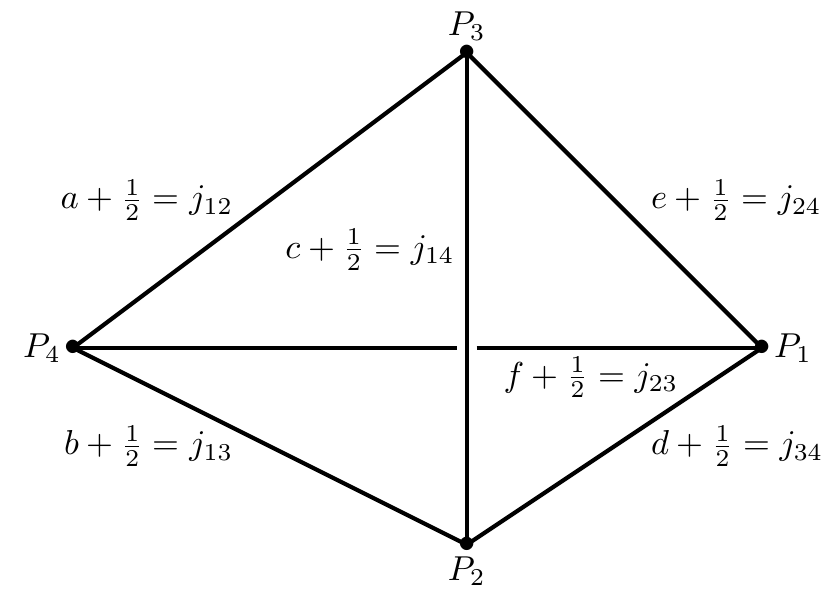}
\caption[]{Tetrahedron associated with the $6j$-symbol $\begin{Bmatrix} a & b & c \\ d & e & f \end{Bmatrix}$.}
\label{6j-Symbol}
\end{center}
\end{figure}

In 1968, Tullio Regge studied with Giorgio Ponzano the large spins limit of the $6j$-symbols defined by Wigner \cite{PR1968}, which have the symmetries of a tetrahedron. On Figure \ref{6j-Symbol}, the labels of the edges of the tetrahedron are indices of $\mathrm{SU}\!\left(2\right)$ representations (the so-called spins). In their work, Ponzano and Regge showed that if we regard those indices as the lengths of the edges of the tetrahedron, then in the large spins limit, the $6j$-symbol associated with the labeling representations is related to the volume $V$ of the tetrahedron according to
\begin{equation}
\begin{Bmatrix}
a & b & c \\
d & e & f
\end{Bmatrix}
\sim\frac{1}{\sqrt{12\pi V}}\cos\left(\sum\limits_{k,l}j_{kl}\theta_{kl}+\frac{\pi}{4}\right),
\end{equation}
$\theta_{kl}$ being the dihedral angles of the tetrahedron.

Studying then the case of $3nj$-symbols, Ponzano and Regge considered triangulated polyhedra, more general than tetrahedra, with edges labeled by indices of $\mathrm{SU}\!\left(2\right)$ representations. Conducting a formal computation consisting of summing over all the possible labels and taking the high spins limit, they recognized the discretization of the partition function of Euclidean General Relativity in dimension $3$ for an empty space that Regge had studied in an earlier work in which he expressed General Relativity without coordinates but using a cellular decomposition of the space-time \cite{REG1961}. Unfortunately, both quantities are ill-defined. Indeed, the sum over all representations of $\mathrm{SU}\!\left(2\right)$ is infinite, and the partition function too, as it consists of a sum over the infinite dimensional space of all the cellular decompositions of space-time.

Those ideas slept nearly 20 years before emerging back in 1992, when Vladimir Turaev and Oleg Viro defined the so-called state-sum invariant of $3$-manifolds, or Turaev-Viro (TV) invariant, using the same idea as Ponzano and Regge of labeled edges of a triangulation of the manifold considered \cite{TV1992}. This construction was then rewritten in the formalism of spherical categories by John Barrett and Bruce Westbury \cite{BW1996}. In this framework, it is possible to construct the invariant with a finite set of objects, according to the so-called ``domination axiom'' that the category has to satisfy. The expression of the TV invariant is then formally very similar to the formula of Ponzano and Regge in the case of a realization of the category by representations of $\mathcal{U}_{q}\!\left(\mathfrak{sl}_{2}\!\left(\mathbb{C}\right)\right)$ (\textit{i.e.} a quantum deformation of the universal enveloping algebra of $\mathfrak{sl}_{2}\!\left(\mathbb{C}\right)$ at a root of unity $q$). Contrary to the formula of Ponzano and Regge, the TV invariant is a well-defined quantity, since the domination axiom imposes the sum to be taken over a period of the cyclic representations of $\mathcal{U}_{q}\!\left(\mathfrak{sl}_{2}\!\left(\mathbb{C}\right)\right)$. It can be therefore considered as a regularization of the formula of Ponzano and Regge.

More recently, in the early years 2000, these works appeared to be particularly useful to regularize some kind of $\mathrm{SU}\!\left(2\right)$ BF theory showing up naturally in Loop Quantum Gravity, see \textit{e.g.} the works of Laurent Freidel, David Louapre and Etera Livine \cite{FL2004, FL2005, FL2006}.

A first question we may ask is, what would an Abelian version of the TV construction look like, and would it be related in some manner to a $\mathrm{U}\!\left(1\right)$ version of the BF theory\footnote{Remark that $\mathrm{U}\!\left(1\right)$ might be regarded as simpler than $\mathrm{SU}\!\left(2\right)$ since $\mathrm{U}\!\left(1\right)$ is Abelian, but $\mathrm{U}\!\left(1\right)$ is non-simply connected, contrary to $\mathrm{SU}\!\left(2\right)$. Hence, there are non-trivial isomorphism classes of $\mathrm{U}\!\left(1\right)$ bundles while $\mathrm{SU}\!\left(2\right)$ bundles are all trivializable.}? These questions were answered between 2015 and 2017 in \cite{MT2016JMP,MT2016NPB,MT2017}. This investigation was purely $3$-dimensional, but the Abelian TV construction can easily be generalized to any dimension. Likewise, the $\mathrm{U}\!\left(1\right)$ BF theory can also be generalized to any dimension. Hence, a natural question is: Is the relation between the Abelian TV invariant and the $\mathrm{U}\!\left(1\right)$ BF theory true in any dimension? This question is answered in \cite{HMT2022_1,HMT2022_2} and the main points of the investigation are presented in this proceeding. 

We will first explain what a ``higher $\mathrm{U}\!\left(1\right)$ BF theory'', and second, what a ``higher Abelian TV invariant'' consist of. Third, we will show how this ``higher Abelian TV invariant'' can be understood as a discrete BF theory. Finally, we will explain how we can make sense of the claim ``BF and TV theories are the same'' in this higher dimensional Abelian case.

\subsection{A bit of vocabulary}

In this proceeding, we will use the following notations and terminology:
\begin{itemize}
\item[-] $\mathcal{F}_{M}$ for the ``space of fields'' over $M$, \textit{e.g.} $C^{\infty}\left(M\right)$ (scalar fields), $\Gamma\left(M,E\right)$ (sections of a vector bundle $E$ over $M$), $\mathrm{Conn}_{G}\!\left(M\right)$...
\item[-] $G$ for the ``gauge group'' (compact Lie group), \textit{e.g.} most of the time $G = \mathrm{SU}\!\left(N\right)$ or $\mathrm{U}\!\left(1\right)$,
\item[-] $\mathcal{G} \subseteq \mathrm{Map}\left(M,G\right)$, which acts on $\mathcal{F}_{M}$,
\item[-] $S:\mathcal{F}_{M}\to\R$ for the ``action functional'' of the theory considered.
\item[-] If $S$ is $\mathcal{G}$-invariant, \textit{i.e.} if $\forall\,\phi\in\mathcal{F}_{M},\,\forall\,g\in\mathcal{G},\,S\left(g\cdot\phi\right) = S\left(\phi\right)$, then $S$ defines a ``gauge theory''. 
\item[-] If $S:\mathcal{F}_{M}\to\R$ defines a gauge theory, then we may try to consider rather $S:\slfrac{\mathcal{F}_{M}}{\mathcal{G}}\to\R$. 
\item[-] If $e^{2\pi iS}$ is $\mathcal{G}$-invariant, \textit{i.e.} if $\forall\,\phi\in\mathcal{F}_{M},\,\forall\,g\in\mathcal{G},\,S\left(g\cdot\phi\right) = S\left(\phi\right) + n$, $n\in\Z$, then $S$ defines a ``quantum gauge theory''. 
\item[-] If $S:\mathcal{F}_{M}\to\R$ defines a quantum gauge theory, then we may try to consider rather $S:\slfrac{\mathcal{F}_{M}}{\mathcal{G}}\to\slfrac{\R}{\Z}$.
\end{itemize}

\section{Generalized $\mathrm{U}\!\left(1\right)$ BF theory}

\subsection{The standard $\mathrm{SU}\!\left(2\right)$ case}

In this subsection, $M^{\left(3\right)}$ is a closed connected oriented smooth $3$-manifold. Remark that, since $\mathrm{SU}\!\left(2\right)$ is simply connected, then any $\mathrm{SU}\!\left(2\right)$ principal bundle over $M^{\left(3\right)}$ is necessarily trivializable. 

The first field of the standard $\mathrm{SU}\!\left(2\right)$ 3D BF theory is an $\mathrm{SU}\!\left(2\right)$-connection $A^{\left(1\right)}$. Since any $\mathrm{SU}\!\left(2\right)$ principal bundle over $M^{\left(3\right)}$ is trivializable,  $A^{\left(1\right)}$ can be regarded as a (global) $1$-form with coefficients in $\mathfrak{su}\!\left(2\right)$ (the upper index keeps track of the form degree of the object) such that the group $\mathcal{G}$ acts on $A^{\left(1\right)}$ as
\begin{equation}
g\cdot A^{\left(1\right)} := g^{-1}A^{\left(1\right)}g + g^{-1}\opd g,
\end{equation}
for $g\in\mathcal{G}$ and $A^{\left(1\right)}\in\mathrm{Conn}_{\mathrm{SU}\!\left(2\right)}\!\left(M^{\left(3\right)}\right)$.

Recall that the curvature of $A^{\left(1\right)}$ is 
\begin{equation}
F^{\left(2\right)}\left(A^{\left(1\right)}\right) 
= \opd A^{\left(1\right)} + A^{\left(1\right)}\wedge A^{\left(1\right)}
\in\Omega^{2}\!\left(M^{\left(3\right)},\mathfrak{su}\!\left(2\right)\right)
\end{equation}
and $F^{\left(2\right)}\left(A^{\left(1\right)}\right)$ transforms under the action of $\mathcal{G}$ on $A^{\left(1\right)}$ as
\begin{equation}
F^{\left(2\right)}\left(g\cdot A^{\left(1\right)}\right) = g^{-1}F^{\left(2\right)}\left(A^{\left(1\right)}\right)g,
\end{equation}
for $g\in\mathcal{G}$ and $A^{\left(1\right)}\in\mathrm{Conn}_{\mathrm{SU}\!\left(2\right)}\!\left(M^{\left(3\right)}\right)$.

The second field of the standard $\mathrm{SU}\!\left(2\right)$ 3D BF theory is a (global) $1$-form $B^{\left(1\right)}$ with coefficients in $\mathfrak{su}\!\left(2\right)$ such that the group $\mathcal{G}$ acts on $B^{\left(1\right)}$ as
\begin{equation}
g\cdot B^{\left(1\right)} = g^{-1}B^{\left(1\right)}g,
\end{equation}
for $g\in\mathcal{G}$ and $B^{\left(1\right)}\in\Omega^{1}\!\left(M^{\left(3\right)},\mathfrak{su}\!\left(2\right)\right)$.

Then the group $\mathcal{G}$ acts on the pair $\left(A^{\left(1\right)}, B^{\left(1\right)}\right)\in\mathrm{Conn}_{\mathrm{SU}\!\left(2\right)}\!\left(M^{\left(3\right)}\right)\times\Omega^{1}\!\left(M^{\left(3\right)},\mathfrak{su}\!\left(2\right)\right)$ diagonally, \textit{i.e.}
\begin{equation}
g\cdot\left(A^{\left(1\right)}, B^{\left(1\right)}\right) 
= \left(g\cdot A^{\left(1\right)}, g\cdot B^{\left(1\right)}\right)
= \left(g^{-1}A^{\left(1\right)}g + g^{-1}\opd g, g^{-1}B^{\left(1\right)}g\right),
\end{equation}
for $g\in\mathcal{G}$ and $\left(A^{\left(1\right)}, B^{\left(1\right)}\right)\in\mathrm{Conn}_{\mathrm{SU}\!\left(2\right)}\!\left(M^{\left(3\right)}\right)\times\Omega^{1}\!\left(M^{\left(3\right)},\mathfrak{su}\!\left(2\right)\right)$.

\begin{definition}
The standard $\mathrm{SU}\!\left(2\right)$ 3D BF theory with a coupling constant $k$ is defined by
\begin{equation}
\mathcal{F}_{M^{\left(3\right)}}
= \mathrm{Conn}_{\mathrm{SU}\!\left(2\right)}\!\left(M^{\left(3\right)}\right)
\times\Omega^{1}\!\left(M^{\left(3\right)},\mathfrak{su}\!\left(2\right)\right)
\end{equation}
as space of fields over which is defined the action functional\footnote{Whence the name ``BF theory'' which is \textit{not} an acronyme, but stands for ``$B\wedge F$''.}
\begin{align}
S_{\mathrm{BF}_{k}}\!\left(A^{\left(1\right)},B^{\left(1\right)}\right) 
= k\displaystyle\int_{M^{\left(3\right)}}\trace{B^{\left(1\right)}\wedge F^{\left(2\right)}\left(A^{\left(1\right)}\right)}.
\end{align}
\end{definition}

\begin{proposition}
This action functional defines a gauge theory, \textit{i.e.}  
\begin{align}
&S_{\mathrm{BF}_{k}}\!\left(g\cdot\left(A^{\left(1\right)}, B^{\left(1\right)}\right)\right) 
= S_{\mathrm{BF}_{k}}\!\left(A^{\left(1\right)},B^{\left(1\right)}\right),
\end{align}
for $g\in\mathcal{G}$, $A^{\left(1\right)}\in\mathrm{Conn}_{\mathrm{SU}\!\left(2\right)}\!\left(M^{\left(3\right)}\right)$ and $B^{\left(1\right)}\in\Omega^{1}\!\left(M^{\left(3\right)},\mathfrak{su}\!\left(2\right)\right)$.
\end{proposition}

For interesting facts about the $\mathrm{SU}\!\left(2\right)$ BF theory, in particular its relation to Alexander-Conway invariant of knots and to the $\mathrm{SU}\!\left(2\right)$ Chern-Simons theory, see \cite{CCRM1995,CCRFM1998}. 

\begin{remark}
In this $\mathrm{SU}\!\left(2\right)$ case, $k$ is \textit{not} quantized, \textit{i.e.} $k$ can be any real number.
\end{remark}

\subsection{The 3D $\mathrm{U}\!\left(1\right)$ case}

In this subsection, $M^{\left(3\right)}$ is still a closed connected oriented smooth $3$-manifold.

The $\mathrm{U}\!\left(1\right)$ case is fundamentally different, as there are in general isomorphism classes of non-trivial bundles\footnote{Recall that $\mathrm{U}\!\left(1\right)$ principal bundles are classified by the second \v{C}ech cohomology $\check{H}^{2}\left(M^{\left(3\right)},\Z\right)$.}. The $\mathrm{U}\!\left(1\right)$ case \textit{cannot} be simply regarded as the abelianization of the $\mathrm{SU}\!\left(2\right)$ case.

Moreover, unlike the $\mathrm{SU}\!\left(2\right)$ case, we want here our fields $\mathbf{A}^{\left(1\right)}$ and $\mathbf{B}^{\left(1\right)}$ to be \textit{gauge classes} of $\mathrm{U}\!\left(1\right)$ connections (whence the bold face used for denoting the fields), which means we won't go through the usual gauge fixing procedure \cite{MNE2019}. 

Fortunately, it turns out that the space of gauge classes of $\mathrm{U}\!\left(1\right)$ connections is perfectly well-known: This is the first group of Deligne-Beilinson cohomology $H^{1}_{\mathrm{DB}}\left(M^{\left(3\right)},\Z\right)$, also known as second group of differential cohomology $\widehat{H}^{2}\left(M^{\left(3\right)},\Z\right)$\footnote{The convention of degree of $H^{1}_{\mathrm{DB}}\left(M^{\left(3\right)},\Z\right)$ allows to keep track of the form degree of the field.}. For a description and more references on this topic, see \cite{BB2014,BGST2005}. We recall here a few important facts we will need later on in this paper.

\begin{proposition}
The group $H^{1}_{\mathrm{DB}}\left(M^{\left(3\right)},\Z\right)$ is described by the following short exact sequence that splits:
\begin{equation}
\label{SES_1}
0
\to \slfrac{\Omega^{1}\left(M^{\left(3\right)}\right)}{\Omega^{1}_{\Z}\left(M^{\left(3\right)}\right)}
\to H^{1}_{\mathrm{DB}}\left(M^{\left(3\right)},\Z\right)
\to \check{H}^{2}\left(M^{\left(3\right)},\Z\right)
\to 0
\end{equation}
where $\Omega^{1}\left(M^{\left(3\right)}\right)$ is the space of $1$-forms over $M^{\left(3\right)}$, $\Omega^{1}_{\Z}\left(M^{\left(3\right)}\right)$ is the space of closed $1$-forms with integral periods over $M^{\left(3\right)}$, and $\check{H}^{2}\left(M^{\left(3\right)},\Z\right)$ is the second \v{C}ech cohomology. 
\end{proposition}

\begin{proposition}
Over $H^{1}_{\mathrm{DB}}\left(M^{\left(3\right)},\Z\right)$, there exists 
\begin{itemize}
\item[-] A (symmetric) pairing
\begin{equation}
\star : \left. 
\begin{tabular}{ccc}
$H^{1}_{\mathrm{DB}}\left(M^{\left(3\right)},\Z\right)\times H^{1}_{\mathrm{DB}}\left(M^{\left(3\right)},\Z\right)$
& $\longrightarrow$ 
& $\slfrac{\Omega^{3}\left(M^{\left(3\right)}\right)}{\Omega^{3}_{\Z}\left(M^{\left(3\right)}\right)}$ \\
$\left(\mathbf{A}^{\left(1\right)},\mathbf{B}^{\left(1\right)}\right)$ 
& $\longmapsto$
& $\mathbf{A}^{\left(1\right)}\star\mathbf{B}^{\left(1\right)}$
\end{tabular}
\right.
\end{equation}
\item[-] A notion of integral over singular $1$-cycles
\begin{equation}
\int : \left.
\begin{tabular}{ccc}
$Z_{1}\left(M^{\left(3\right)},\Z\right)\times H^{1}_{\mathrm{DB}}\left(M^{\left(3\right)},\Z\right)$
& $\longrightarrow$ 
& $\slfrac{\R}{\Z}$ \\
& \\
$\left(z_{\left(1\right)},\mathbf{A}^{\left(1\right)}\right)$ 
& $\longmapsto$
& $\int_{z_{\left(1\right)}}\mathbf{A}^{\left(1\right)}$
\end{tabular}
\right.
\end{equation}
\end{itemize}
\end{proposition}

\begin{proposition}
Another notion of integral will be very useful here. This is
\begin{equation}
\label{Int_M3}
\int_{M^{\left(3\right)}} : \left.
\begin{tabular}{ccc}
$\slfrac{\Omega^{3}\left(M^{\left(3\right)}\right)}{\Omega^{3}_{\Z}\left(M^{\left(3\right)}\right)}$
& $\longrightarrow$ 
& $\slfrac{\R}{\Z}$ \\
& \\
$\mathbf{L}^{\left(3\right)}$ 
& $\longmapsto$
& $\int_{M^{\left(3\right)}} \mathbf{L}^{\left(3\right)}$
\end{tabular}
\right.
\end{equation}
\end{proposition}

\begin{remark}
We want to emphasize that the integrals introduced in the two previous propositions are $\slfrac{\R}{\Z}$-valued. Said differently, only the complex exponential truly makes sense, and we will see later on that this is anyway what we want to consider in a quantum treatment of the theory for the partition function and the expectation values of observables.
\end{remark}

\begin{definition}
The $\mathrm{U}\!\left(1\right)$ 3D BF theory with a coupling constant $k\in\Z$\footnote{The coupling constant $k$ has to be an integer since the integral in the action is $\slfrac{\R}{\Z}$-valued.} is defined by 
\begin{equation}
\slfrac{\mathcal{F}_{M^{\left(3\right)}}}{\mathcal{G}} 
= H^{1}_{\mathrm{DB}}\left(M^{\left(3\right)},\Z\right)
\times H^{1}_{\mathrm{DB}}\left(M^{\left(3\right)},\Z\right)
\end{equation}
as space of fields (modulo gauge transformations) over which is defined the \textit{quantum} action functional
\begin{equation}
S_{\mathrm{BF}_{k}}\left(\mathbf{A}^{\left(1\right)},\mathbf{B}^{\left(1\right)}\right)
\underset{\Z}{=} k\displaystyle\int_{M^{\left(3\right)}}\mathbf{B}^{\left(1\right)}\star\mathbf{A}^{\left(1\right)} 
\end{equation}
and the \textit{quantum} observables
\begin{equation}
W\!\left(\mathbf{A}^{\left(1\right)}, z^{\mathbf{A}}_{\left(1\right)}, \mathbf{B}, z^{\mathbf{B}}_{\left(1\right)}\right)
= e^{2\pi i\int_{z^{\mathbf{A}}_{\left(1\right)}}\mathbf{A}^{\left(1\right)}}
e^{2\pi i\int_{z^{\mathbf{B}}_{\left(1\right)}}\mathbf{B}^{\left(1\right)}},
\end{equation}
where $\mathbf{A}^{\left(1\right)}, \mathbf{B}^{\left(1\right)}\in H^{1}_{\mathrm{DB}}\left(M^{\left(3\right)},\Z\right)$, $z^{\mathbf{A}}_{\left(1\right)},z^{\mathbf{B}}_{\left(1\right)}\in Z_{1}\left(M^{\left(3\right)},\Z\right)$.  
\end{definition}

\begin{remark}
If $\mathbf{B}^{\left(1\right)} = \mathbf{A}^{\left(1\right)}$, which is possible because we chose $\mathbf{A}^{\left(1\right)}$ and $\mathbf{B}^{\left(1\right)}$ to be both (gauge classes of) $\mathrm{U}\!\left(1\right)$ connections, then we get the quantum action of the $\mathrm{U}\!\left(1\right)$ Chern-Simons theory
\begin{equation}
S_{\mathrm{CS}_{k}}\left(\mathbf{A}^{\left(1\right)}\right) 
= k\displaystyle\int_{M^{\left(3\right)}}\mathbf{A}^{\left(1\right)}\star\mathbf{A}^{\left(1\right)}.
\end{equation}
\end{remark}

\subsection{The nD $\mathrm{U}\!\left(1\right)$ case}

In this subsection, $M^{\left(n\right)}$ is a closed connected oriented smooth $n$-manifold.

\begin{proposition}
The short exact sequence \eqref{SES_1} can actually be generalized as
\begin{equation}
\label{SES_p}
0
\to \slfrac{\Omega^{p}\left(M^{\left(n\right)}\right)}{\Omega^{p}_{\Z}\left(M^{\left(n\right)}\right)}
\to H^{p}_{\mathrm{DB}}\left(M^{\left(n\right)},\Z\right)
\to \check{H}^{p+1}\left(M^{\left(n\right)},\Z\right)
\to 0.
\end{equation}
and it still splits.
\end{proposition}

\begin{remark}
The exact sequence \eqref{SES_p} does make sense for $p = -1$ (provided we set $\slfrac{\Omega^{-1}\left(M^{\left(n\right)}\right)}{\Omega^{-1}_{\Z}\left(M^{\left(n\right)}\right)}\cong 0$, which implies then that $H^{-1}_{\mathrm{DB}}\left(M^{\left(n\right)},\Z\right) \cong \Z$) and $p = n$. However, in the following, we will \textit{not} consider those two cases for the quantum fields of our BF theory, as they would imply that one of the two fields is actually just a number.
\end{remark}

\begin{remark}
The elements of $H^{p}_{\mathrm{DB}}\left(M^{\left(n\right)},\Z\right)$ may be regarded as ``gauge classes of connections over $\left(p-1\right)$-gerbes'', where ``$\left(p-1\right)$-gerbes'' are structures classified up to isomorphism by $\check{H}^{p+1}\left(M^{\left(n\right)},\Z\right)$. Hence, a $\mathrm{U}\!\left(1\right)$ principal bundle is nothing but a $0$-gerbe, and what is usually called as a ``gerbe'' is actually a $1$-gerbe.
\end{remark}

\begin{proposition}
Over $H^{\bullet}_{\mathrm{DB}}\left(M^{\left(n\right)},\Z\right)$, we can define
\begin{itemize}
\item[-] A pairing
\begin{equation}
\star : \left. 
\begin{tabular}{ccc}
$H^{p}_{\mathrm{DB}}\left(M^{\left(n\right)},\Z\right)\times H^{q}_{\mathrm{DB}}\left(M^{\left(n\right)},\Z\right)$
& $\longrightarrow$ 
& $H^{p+q+1}_{\mathrm{DB}}\left(M^{\left(n\right)},\Z\right)$\\
& \\
$\left(\mathbf{A}^{\left(p\right)}, \mathbf{B}^{\left(q\right)}\right)$ 
& $\longmapsto$
& $\mathbf{A}^{\left(p\right)}\star \mathbf{B}^{\left(q\right)}$
\end{tabular}
\right.
\end{equation}
\item[-] A notion of integral over singular $p$-cycles
\begin{equation}
\int : \left.
\begin{tabular}{ccc}
$Z_{p}\left(M^{\left(n\right)},\Z\right)\times H^{p}_{\mathrm{DB}}\left(M^{\left(n\right)},\Z\right)$
& $\longrightarrow$ 
& $\slfrac{\R}{\Z}$ \\
& \\
$\left(z_{\left(p\right)},\mathbf{A}^{\left(p\right)}\right)$ 
& $\longmapsto$
& $\int_{z_{\left(p\right)}}\mathbf{A}^{\left(p\right)}$
\end{tabular}
\right.
\end{equation}
\end{itemize}
\end{proposition}

\begin{remark}
A particular case of this integral is for $p = n$, in which case $z_{\left(n\right)} = M^{\left(n\right)}$ (potentially multiplied by an integer) and
\begin{equation}
H^{n}_{\mathrm{DB}}\left(M^{\left(n\right)},\Z\right)
\cong\slfrac{\Omega^{n}\left(M^{\left(n\right)}\right)}{\Omega^{n}_{\Z}\left(M^{\left(n\right)}\right)}, 
\end{equation}
which is an isomorphism that we implicitly used in \eqref{Int_M3}.
\end{remark}

\begin{remark}
The pairing $\star:H^{p}_{\mathrm{DB}}\left(M^{\left(n\right)},\Z\right)\times H^{q}_{\mathrm{DB}}\left(M^{\left(n\right)},\Z\right)$ is almost graded-commutative, \textit{i.e.} with the convention of degree of $H^{\bullet}_{\mathrm{DB}}\left(M,\Z\right)$, 
\begin{equation}
\mathbf{A}^{\left(p\right)}\star\mathbf{B}^{\left(q\right)} 
= \left(-1\right)^{p+q+1}\mathbf{B}^{\left(q\right)}\star\mathbf{A}^{\left(p\right)} 
\end{equation}
for $\left(\mathbf{A}^{\left(p\right)},\mathbf{B}^{\left(q\right)}\right)\in H^{p}_{\mathrm{DB}}\left(M^{\left(n\right)},\Z\right)\times H^{q}_{\mathrm{DB}}\left(M^{\left(n\right)},\Z\right)$, while with the convention of degree of $\widehat{H}^{\bullet}\left(M^{\left(n\right)},\Z\right)$, the product is truly graded-commutative.
\end{remark}

\begin{definition}
Since we can make sense of Deligne-Beilinson cohomology classes in any degree as well as their product and integral, the $\mathrm{U}\!\left(1\right)$ nD BF theory with a coupling constant $k\in\Z$ will be defined by 
\begin{equation}
\slfrac{\mathcal{F}_{M^{\left(n\right)}}}{\mathcal{G}} 
= H^{p}_{\mathrm{DB}}\left(M^{\left(n\right)},\Z\right)
\times H^{q}_{\mathrm{DB}}\left(M^{\left(n\right)},\Z\right)
\end{equation}
as space of fields (modulo gauge transformations) over which is defined the \textit{quantum} action functional
\begin{equation}
S_{\mathrm{BF}_{k}}\left(\mathbf{A}^{\left(p\right)},\mathbf{B}^{\left(q\right)}\right)
\underset{\Z}{=} k\displaystyle\int_{M^{\left(n\right)}}\mathbf{B}^{\left(q\right)}\star\mathbf{A}^{\left(p\right)} 
\end{equation}
and the \textit{quantum} observables
\begin{equation}
W\!\left(\mathbf{A}^{\left(p\right)}, z^{\mathbf{A}}_{\left(p\right)}, 
\mathbf{B}^{\left(q\right)}, z^{\mathbf{B}}_{\left(q\right)}\right)
= e^{2\pi i\int_{z^{\mathbf{A}}_{\left(p\right)}}\mathbf{A}^{\left(p\right)}}
e^{2\pi i\int_{z^{\mathbf{B}}_{\left(q\right)}}\mathbf{B}^{\left(q\right)}},
\end{equation}
with $n = p + q + 1$, $\left(\mathbf{A}^{\left(p\right)},\mathbf{B}^{\left(q\right)}\right)\in H^{p}_{\mathrm{DB}}\left(M^{\left(n\right)},\Z\right)\times H^{q}_{\mathrm{DB}}\left(M^{\left(n\right)},\Z\right)$, $z^{\mathbf{A}}_{\left(p\right)},z^{\mathbf{B}}_{\left(q\right)}\in Z_{p}\left(M^{\left(n\right)},\Z\right)\times Z_{q}\left(M^{\left(n\right)},\Z\right)$ and $k\in\Z$.
\end{definition}

Exactly like in the 3D case \cite{GT2008,GT2013,GT2014,MT2016JMP,MT2016NPB,MT2017}, quantum computations can be performed using the fact that the short exact sequence \eqref{SES_p} splits, that is, decomposing a Deligne-Beilinson cohomology classes as
\begin{equation}
\mathbf{A}^{\left(p\right)} = \mathbf{A}^{\left(p\right)}_{\mathbf{a}} + \boldsymbol{\alpha}^{\left(p\right)} 
\end{equation}
where the index $\mathbf{a}$ belongs to $\check{H}^{p+1}\left(M^{\left(n\right)},\Z\right)$ and $\boldsymbol{\alpha}^{\left(p\right)}\in\slfrac{\Omega^{p}\left(M^{\left(n\right)}\right)}{\Omega^{p}_{\Z}\left(M^{\left(n\right)}\right)}$, see Figure \ref{Space_Fields}. 

\begin{figure}
\includegraphics{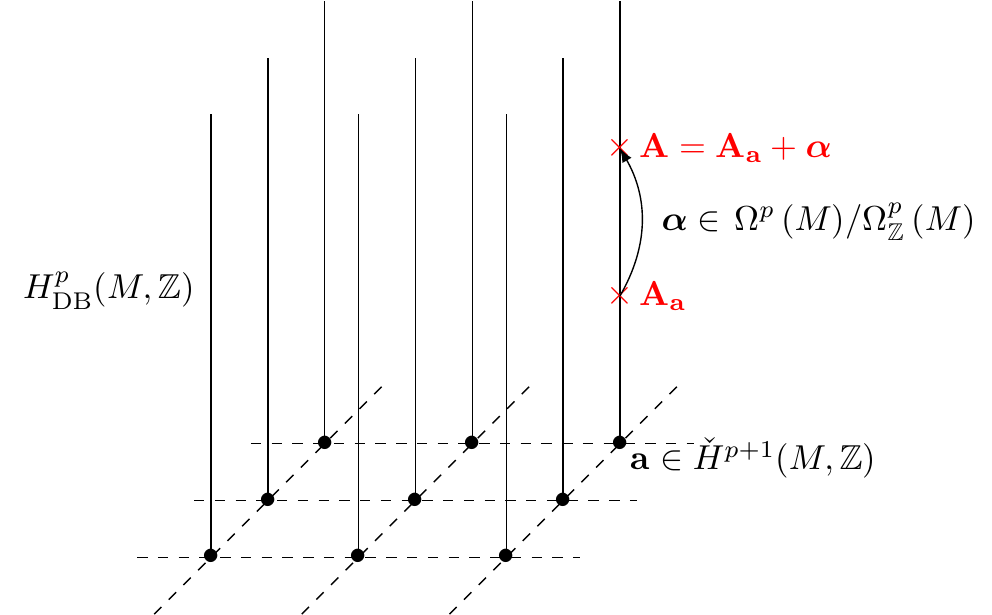}
\caption{Representation of the space of fields $\slfrac{\mathcal{F}_{M^{\left(n\right)}}}{\mathcal{G}} = H^{p}_{\mathrm{DB}}\left(M^{\left(n\right)},\Z\right)$}
\label{Space_Fields}
\end{figure}

We can go even further by writing that, as an Abelian group, 
\begin{equation}
\check{H}^{p+1}\left(M,\Z\right) = \check{F}^{p+1}\oplus\check{T}^{p+1}, 
\end{equation}
where $\check{F}^{p+1}\cong\Z^{b^{p+1}}$ is the free part of $\check{H}^{p+1}\left(M^{\left(n\right)},\Z\right)$, while 
\begin{equation}
\check{T}^{p+1}\cong\slfrac{\Z}{\zeta^{p+1}_{1}\Z}\oplus\hdots\oplus\slfrac{\Z}{\zeta^{p+1}_{t^{p+1}}\Z},
\end{equation}
with $\zeta^{p+1}_{i}\vert \zeta^{p+1}_{i+1}$, is its torsion part, and decomposing $\slfrac{\Omega^{p}\left(M^{\left(n\right)}\right)}{\Omega^{p}_{\Z}\left(M^{\left(n\right)}\right)}$ as 
\begin{align}
&\slfrac{\Omega^{p}\left(M^{\left(n\right)}\right)}{\Omega^{p}_{\Z}\left(M^{\left(n\right)}\right)}\\
\nonumber
&\hspace{1.cm}= \slfrac{\Omega^{p}\left(M^{\left(n\right)}\right)}{\Omega^{p}_{\mathrm{cl}}\left(M^{\left(n\right)}\right)}
\oplus\slfrac{\Omega^{p}_{\mathrm{cl}}\left(M^{\left(n\right)}\right)}{\Omega^{p}_{\Z}\left(M^{\left(n\right)}\right)},
\end{align}
$\Omega^{p}_{\mathrm{cl}}\left(M^{\left(n\right)}\right)$ being the space of closed $p$-forms, knowing that 
\begin{equation}
\slfrac{\Omega^{p}_{\mathrm{cl}}\left(M^{\left(n\right)}\right)}{\Omega^{p}_{\Z}\left(M^{\left(n\right)}\right)} 
\cong \left(\slfrac{\R}{\Z}\right)^{b^{p}}. 
\end{equation}
Hence, we can actually write for any class $\mathbf{A}^{\left(p\right)}\in H^{p}_{\mathrm{DB}}\left(M^{\left(n\right)},\Z\right)$,
\begin{equation}
\mathbf{A}^{\left(p\right)} = \mathbf{A}^{\left(p\right)}_{\mathbf{m}} + \mathbf{A}^{\left(p\right)}_{\boldsymbol{\kappa}} + \boldsymbol{\alpha}^{\left(p\right)}_{0} + \boldsymbol{\alpha}^{\left(p\right)}_{\perp},
\end{equation}
with $\mathbf{m}\in\check{F}^{p+1}$, $\boldsymbol{\kappa}\in\check{T}^{p+1}$, $\boldsymbol{\alpha}^{\left(p\right)}_{\perp}\in\slfrac{\Omega^{p}\left(M^{\left(n\right)}\right)}{\Omega^{p}_{\mathrm{cl}}\left(M^{\left(n\right)}\right)}$ and the so-called ``zero modes'' $\boldsymbol{\alpha}^{\left(p\right)}_{0}\in\slfrac{\Omega^{p}_{\mathrm{cl}}\left(M^{\left(n\right)}\right)}{\Omega^{p}_{\Z}\left(M^{\left(n\right)}\right)}$. 

This decomposition is of course not unique, and the computations we want to perform should be independent from the choice of decomposition. The whole point is to find a nice decomposition that makes our computations easy.

\begin{claim}[\cite{HMT2022_1}, \cite{HMT2022_2}, \cite{GT2008}, \cite{GT2013}, \cite{GT2014}]
There exists a decomposition for $\mathbf{A}^{\left(p\right)}$ and $\mathbf{B}^{\left(q\right)}$ that satisfies the following very convenient properties:
\begin{multicols}{2}
\begin{enumerate}[1)]
\item $\int_{M^{\left(n\right)}}\boldsymbol{\beta}^{\left(q\right)}_{0}\star\mathbf{A}^{\left(p\right)}_{\mathbf{m}_{\mathbf{A}}} 
\underset{\Z}{=} 0$,
\item $\int_{M^{\left(n\right)}}\boldsymbol{\beta}^{\left(q\right)}_{0}\star\mathbf{A}^{\left(p\right)}_{\boldsymbol{\kappa}_{\mathbf{A}}} 
\underset{\Z}{=} 0$,
\item $\int_{M^{\left(n\right)}}\boldsymbol{\beta}^{\left(q\right)}_{\perp}\star\mathbf{A}^{\left(p\right)}_{\boldsymbol{\kappa}_{\mathbf{A}}}
\underset{\Z}{=} 0$,

\item $\int_{M^{\left(n\right)}}\boldsymbol{\beta}^{\left(q\right)}_{0}\star\boldsymbol{\alpha}^{\left(p\right)}_{\perp}\underset{\Z}{=} 0$,
\item $\int_{M^{\left(n\right)}}\boldsymbol{\beta}^{\left(q\right)}_{0}\star\boldsymbol{\alpha}^{\left(p\right)}_{0}\underset{\Z}{=} 0$,
\item $\int_{M^{\left(n\right)}}\mathbf{B}^{\left(q\right)}_{\mathbf{m}_{\mathbf{B}}}
\star\mathbf{A}^{\left(p\right)}_{\mathbf{m}_{\mathbf{A}}} 
\underset{\Z}{=} 0$
\item $\int_{M^{\left(n\right)}}\mathbf{B}^{\left(q\right)}_{\boldsymbol{\kappa}_{\mathbf{B}}}
\star\mathbf{A}^{\left(p\right)}_{\boldsymbol{\kappa}_{\mathbf{A}}} 
\underset{\Z}{=} - Q\left(\boldsymbol{\kappa}_{\mathbf{B}},\boldsymbol{\kappa}_{\mathbf{A}}\right)$
\end{enumerate}
\end{multicols}
\vspace{-0.25cm}
\noindent where $Q:\check{T}^{q+1}\times\check{T}^{p+1}\to\slfrac{\Q}{\Z}$ is a generalization of the so-called ``linking form'' \cite{MW2017}.
\end{claim}

\begin{consequence}
The previous properties imply that
\begin{align}
\nonumber
\int_{M^{\left(n\right)}}\mathbf{B}^{\left(q\right)}\star\mathbf{A}^{\left(p\right)}\underset{\Z}{=} 
&\int_{M^{\left(n\right)}}\boldsymbol{\beta}^{\left(q\right)}_{\perp}\star \mathbf{A}^{\left(p\right)}_{\mathbf{m}_{\mathbf{A}}} 
+ \int_{M^{\left(n\right)}}\mathbf{B}^{\left(q\right)}_{\mathbf{m}_{\mathbf{B}}}\star\boldsymbol{\alpha}^{\left(p\right)}_{\perp}
+ \int_{M^{\left(n\right)}}\boldsymbol{\beta}^{\left(q\right)}_{\perp}\star\boldsymbol{\alpha}^{\left(p\right)}_{\perp}\\
&+ \mathbf{m}_{\mathbf{B}}\cdot\boldsymbol{\alpha}^{\left(p\right)}_{0}
+ \mathbf{m}_{\mathbf{A}}\cdot\boldsymbol{\beta}^{\left(q\right)}_{0}
- Q\left(\boldsymbol{\kappa}_{\mathbf{B}},\boldsymbol{\kappa}_{\mathbf{A}}\right)
\end{align}
\end{consequence}

\textsc{Heuristic and motivational calculation.} We want to study functional integrals such as
\begin{align}
\label{Partition_Function}
\mbox{``}\,\,\mathcal{Z}^{p}_{\mathrm{BF}_{k}}\left(M^{\left(n\right)}\right)
=\frac{1}{\mathscr{N}^{p}_{\mathrm{BF}_{k}}\left(M^{\left(n\right)}\right)}
\int_{H^{q}_{\mathrm{DB}}\left(M^{\left(n\right)},\Z\right)\times H^{p}_{\mathrm{DB}}\left(M^{\left(n\right)},\Z\right)}
\hspace{-3.cm}\mathscr{D}\mathbf{B}^{\left(q\right)}\,\mathscr{D}\mathbf{A}^{\left(p\right)}\, 
e^{2\pi iS_{\mathrm{BF}_{k}}\left(\mathbf{A}^{\left(p\right)},\mathbf{B}^{\left(q\right)}\right)}\,\,\mbox{''}
\end{align} 
(the so-called ``partition function'') and
\begin{align}
\nonumber
\mbox{``}\,\,&\left\langle W\!\left(\mathbf{A}^{\left(p\right)}, z^{\mathbf{A}}_{\left(p\right)}, \mathbf{B}^{\left(q\right)}, z^{\mathbf{B}}_{\left(q\right)}\right)\right\rangle\\
\nonumber
&\hspace{1.cm}=\frac{1}{\mathscr{N}^{p}_{\mathrm{BF}_{k}}\left(M^{\left(n\right)}\right)}
\int_{H^{q}_{\mathrm{DB}}\left(M^{\left(n\right)},\Z\right)\times H^{p}_{\mathrm{DB}}\left(M^{\left(n\right)},\Z\right)}
\hspace{-2.cm}\mathscr{D}\mathbf{B}^{\left(q\right)}\,\mathscr{D}\mathbf{A}^{\left(p\right)}\\
&\hspace{4.5cm}W\!\left(\mathbf{A}^{\left(p\right)}, z^{\mathbf{A}}_{\left(p\right)}, 
\mathbf{B}^{\left(q\right)}, z^{\mathbf{B}}_{\left(q\right)}\right)\,
e^{2\pi iS_{\mathrm{BF}_{k}}\left(\mathbf{A}^{\left(p\right)},\mathbf{B}^{\left(q\right)}\right)}\,\,\mbox{''}
\end{align} 
(the so-called ``expectation value of observables'') which are ill-defined integrals.

We will focus here on the partition function \eqref{Partition_Function} and try to propose an appropriate definition. If we use in this integral the decomposition and the properties we recalled above, we can write
\begin{align}
\nonumber
\mbox{``}\,\,&\mathcal{Z}^{p}_{\mathrm{BF}_{k}}\left(M^{\left(n\right)}\right)\\
\nonumber
&= \frac{1}{\mathscr{N}_{\mathrm{BF}_{k}}\left(M^{\left(n\right)}\right)}
\sum_{\boldsymbol{\kappa}_{\mathbf{B}},\boldsymbol{\kappa}_{\mathbf{A}}\in\check{T}^{q+1}\times\check{T}^{p+1}}
e^{-2\pi ik Q\left(\boldsymbol{\kappa}_{\mathbf{B}},\boldsymbol{\kappa}_{\mathbf{A}}\right)}\\
&\hspace{0.5cm}\sum_{\mathbf{m}_{\mathbf{B}},\mathbf{m}_{\mathbf{A}}\in\check{F}^{q+1}\times\check{F}^{p+1}}
\int_{\left(\slfrac{\R}{\Z}\right)^{b^{q}}\times\left(\slfrac{\R}{\Z}\right)^{b^{p}}}
\hspace{-.1cm}\opd^{b^{q}}\boldsymbol{\beta}_{0}\,
\opd^{b^{p}}\boldsymbol{\alpha}_{0}\,
e^{2\pi ik\left(\int_{M^{\left(n\right)}}
\mathbf{m}_{\mathbf{B}}\cdot\boldsymbol{\alpha}^{\left(p\right)}_{0}
+ \mathbf{m}_{\mathbf{A}}\cdot\boldsymbol{\beta}^{\left(q\right)}_{0}\right)}\\
\nonumber
&\hspace{0.5cm}\int_{\slfrac{\Omega^{q}\left(M^{\left(n\right)}\right)}{\Omega^{q}_{\mathrm{cl}}\left(M^{\left(n\right)}\right)}
\times \slfrac{\Omega^{p}\left(M^{\left(n\right)}\right)}{\Omega^{p}_{\mathrm{cl}}\left(M^{\left(n\right)}\right)}}
\mathscr{D}\boldsymbol{\beta}^{\left(q\right)}_{\perp}\,\mathscr{D}\boldsymbol{\alpha}^{\left(p\right)}_{\perp}\\
\nonumber
&\hspace{1.5cm}e^{2\pi ik\left(\int_{M^{\left(n\right)}}
\boldsymbol{\beta}^{\left(q\right)}_{\perp}\star\mathbf{A}^{\left(p\right)}_{\mathbf{m}_{\mathbf{A}}} 
+ \int_{M^{\left(n\right)}}\mathbf{B}^{\left(q\right)}_{\mathbf{m}_{\mathbf{B}}}\star\boldsymbol{\alpha}^{\left(p\right)}_{\perp}
+ \int_{M^{\left(n\right)}}\boldsymbol{\beta}^{\left(q\right)}_{\perp}\star\boldsymbol{\alpha}^{\left(p\right)}_{\perp}\right)}
\,\,\mbox{''}
\end{align} 
and the integral over the topological sector $\left(\slfrac{\R}{\Z}\right)^{b^{q}}\times\left(\slfrac{\R}{\Z}\right)^{b^{p}}$ can be performed and gives
\begin{align}
\nonumber
\mbox{``}\,\,&\mathcal{Z}^{p}_{\mathrm{BF}_{k}}\left(M^{\left(n\right)}\right)\\
\nonumber
&= \frac{1}{\mathscr{N}_{\mathrm{BF}_{k}}\left(M^{\left(n\right)}\right)}
\sum_{\boldsymbol{\kappa}_{\mathbf{B}},\boldsymbol{\kappa}_{\mathbf{A}}\in\check{T}^{q+1}\times\check{T}^{p+1}}
e^{-2\pi ik Q\left(\boldsymbol{\kappa}_{\mathbf{B}},\boldsymbol{\kappa}_{\mathbf{A}}\right)}\\
&\hspace{0.5cm}\int_{\slfrac{\Omega^{q}\left(M^{\left(n\right)}\right)}{\Omega^{q}_{\mathrm{cl}}\left(M^{\left(n\right)}\right)}
\times \slfrac{\Omega^{p}\left(M^{\left(n\right)}\right)}{\Omega^{p}_{\mathrm{cl}}\left(M^{\left(n\right)}\right)}}
\mathscr{D}\boldsymbol{\beta}^{\left(q\right)}_{\perp}\,\mathscr{D}\boldsymbol{\alpha}^{\left(p\right)}_{\perp}
e^{2\pi ik\int_{M^{\left(n\right)}}\boldsymbol{\beta}^{\left(q\right)}_{\perp}\star\boldsymbol{\alpha}^{\left(p\right)}_{\perp}}\\
\nonumber
&\hspace{0.5cm}\sum_{\mathbf{m}_{\mathbf{B}},\mathbf{m}_{\mathbf{A}}\in F^{q+1}\times F^{p+1}}
\delta_{\mathbf{m}_{B}}\delta_{\mathbf{m}_{A}} e^{2\pi ik\left(\int_{M^{\left(n\right)}}
\boldsymbol{\beta}^{\left(q\right)}_{\perp}\star\mathbf{A}^{\left(p\right)}_{\mathbf{m}_{\mathbf{A}}} 
+ \int_{M^{\left(n\right)}}\mathbf{B}^{\left(q\right)}_{\mathbf{m}_{\mathbf{B}}}\star\boldsymbol{\alpha}^{\left(p\right)}_{\perp}\right)}
\,\,\mbox{''}
\end{align} 
$\delta_{\mathbf{m}_{B}}$ and $\delta_{\mathbf{m}_{A}}$ being regular Kronecker symbols. Then, we can sum over the topological sector $\check{F}^{q+1}\times\check{F}^{p+1}$, and we obtain
\begin{align}
\nonumber
\mbox{``}\,\,&\mathcal{Z}^{p}_{\mathrm{BF}_{k}}\left(M^{\left(n\right)}\right)\\
\nonumber
&= \frac{1}{\mathscr{N}_{\mathrm{BF}_{k}}\left(M^{\left(n\right)}\right)}
\sum_{\boldsymbol{\kappa}_{\mathbf{B}},\boldsymbol{\kappa}_{\mathbf{A}}\in\check{T}^{q+1}\times\check{T}^{p+1}}
e^{-2\pi ik Q\left(\boldsymbol{\kappa}_{\mathbf{B}},\boldsymbol{\kappa}_{\mathbf{A}}\right)}\\
&\hspace{0.5cm}\int_{\slfrac{\Omega^{q}\left(M^{\left(n\right)}\right)}{\Omega^{q}_{\mathrm{cl}}\left(M^{\left(n\right)}\right)}
\times \slfrac{\Omega^{p}\left(M^{\left(n\right)}\right)}{\Omega^{p}_{\mathrm{cl}}\left(M^{\left(n\right)}\right)}}
\mathscr{D}\boldsymbol{\beta}^{\left(q\right)}_{\perp}\,\mathscr{D}\boldsymbol{\alpha}^{\left(p\right)}_{\perp}
e^{2\pi ik\int_{M^{\left(n\right)}}\boldsymbol{\beta}^{\left(q\right)}_{\perp}\star\boldsymbol{\alpha}^{\left(p\right)}_{\perp}}
\,\,\mbox{''}
\end{align}
We observe here a convenient full decoupling of the two remaining topological sectors, $\slfrac{\Omega^{q}\left(M^{\left(n\right)}\right)}{\Omega^{q}_{\mathrm{cl}}\left(M^{\left(n\right)}\right)}\times\slfrac{\Omega^{p}\left(M^{\left(n\right)}\right)}{\Omega^{p}_{\mathrm{cl}}\left(M^{\left(n\right)}\right)}$ and $\check{T}^{q+1}\times\check{T}^{p+1}$. The contribution of the topological sector $\slfrac{\Omega^{q}\left(M^{\left(n\right)}\right)}{\Omega^{q}_{\mathrm{cl}}\left(M^{\left(n\right)}\right)}\times \slfrac{\Omega^{p}\left(M^{\left(n\right)}\right)}{\Omega^{p}_{\mathrm{cl}}\left(M^{\left(n\right)}\right)}$ is infinite dimensional and we choose to eliminate it\footnote{A lot of authors \cite{SCH1978,BBRT1991, AS1991} extract from this part the Reidemeister torsion of $M^{\left(n\right)}$. However, they do not work with the gauge classes of fields. They fix the gauge, and for that they usually introduce a metric, which they want afterwards to get rid of, as the theory is expected to be topological, \textit{i.e.} partition function and expectation values of observables should not depend on any metric.} thanks to the normalization, \textit{i.e.} by writing formally
\begin{align}
\nonumber
\mbox{``}\,\,&\mathscr{N}^{p}_{\mathrm{BF}_{k}}\left(M^{\left(n\right)}\right)\\
&= \int_{\slfrac{\Omega^{q}\left(M^{\left(n\right)}\right)}{\Omega^{q}_{\mathrm{cl}}\left(M^{\left(n\right)}\right)}
\times\slfrac{\Omega^{p}\left(M^{\left(n\right)}\right)}{\Omega^{p}_{\mathrm{cl}}\left(M^{\left(n\right)}\right)}}
\hspace{-0.25cm}\mathscr{D}\boldsymbol{\beta}^{\left(q\right)}_{\perp}\,\mathscr{D}\boldsymbol{\alpha}^{\left(p\right)}_{\perp}
e^{2\pi ik\int_{M^{\left(n\right)}}\boldsymbol{\beta}^{\left(q\right)}_{\perp}\star\boldsymbol{\alpha}^{\left(p\right)}_{\perp}}.
\,\,\mbox{''}
\label{Normalization}
\end{align}
Therefore, only remains the contribution of the sector $\check{T}^{q+1}\times\check{T}^{p+1}$ which we will regard as our \textit{definition} of the partition function $\mathcal{Z}^{p}_{\mathrm{BF}_{k}}\left(M^{\left(n\right)}\right)$.

\begin{definition}
In the following,
\begin{equation}
\mathcal{Z}^{p}_{\mathrm{BF}_{k}}\left(M^{\left(n\right)}\right) 
:= \sum_{\boldsymbol{\kappa}_{B},\boldsymbol{\kappa}_{A}\in\check{T}^{q+1}\times\check{T}^{p+1}}
e^{-2\pi i kQ\left(\boldsymbol{\kappa}_{B},\boldsymbol{\kappa}_{A}\right)}
\end{equation}
\end{definition}

\begin{proposition}
For the lens space $L\left(r,s\right)$ in dimension $3$, we have
\begin{equation}
\mathcal{Z}^{1}_{\mathrm{BF}_{k}}\left(L\left(r,s\right)\right)
=\mathrm{gcd}\!\left(r,k\right)r
= \abs{\check{T}^{2}}\abs{\Hom{\check{T}^{2}}{\slfrac{\Z}{k\Z}}}
\end{equation}
\end{proposition}

\begin{proof}
For a lens space $L\left(r,s\right)$, we know that $Q = \frac{r}{s}$, so
\begin{align}
\mathcal{Z}^{1}_{\mathrm{BF}_{k}}\left(L\left(r,s\right)\right) 
= \sum_{n,m = 0}^{r-1} e^{-2\pi i k\frac{s}{r}mn}
= \sum_{n,m = 0}^{r-1} e^{-2\pi i k'\frac{s}{r'}mn}
\end{align}
where $k = k'\mathrm{gcd}\!\left(r,k\right)$ and $r = r'\mathrm{gcd}\!\left(r,k\right)$, then
\begin{align}
\mathcal{Z}^{1}_{\mathrm{BF}_{k}}\left(L\left(r,s\right)\right) 
=& \sum_{n = 0}^{r-1} \left(\sum_{m = 0}^{r-1}e^{-2\pi i k'\frac{s}{r'}n}\right)^{m}
\end{align}
and the inner parenthesis is $0$ whenever $r'\nmid n$ and $r$ whenever $r'\mid n$, whence
\begin{align}
\mathcal{Z}^{1}_{\mathrm{BF}_{k}}\left(M^{\left(n\right)}\right)
=\mathrm{gcd}\!\left(r,k\right)r
=\abs{\check{T}^{2}}\abs{\Hom{\check{T}^{2}}{\slfrac{\Z}{k\Z}}}
\end{align}
\end{proof}

\begin{proposition}
More generally, using the properties of $Q$, we can show that
\begin{equation}
\mathcal{Z}^{p}_{\mathrm{BF}_{k}}\left(M^{\left(n\right)}\right) 
= \prod_{i=1}^{t^{p+1}}\mathrm{gcd}\!\left(\zeta^{p+1}_{i},k\right)\zeta^{p+1}_{i}
= \abs{\check{T}^{p+1}}\abs{\Hom{\check{T}^{p+1}}{\slfrac{\Z}{k\Z}}}
\end{equation}
\end{proposition}

\begin{remark}
The \v{C}ech homology (respectively cohomology) of $M^{\left(n\right)}$ is isomorphic to the simplicial homology (respectively cohomology) of $M^{\left(n\right)}$, its singular homology (respectively cohomology) and its cellular homology (respectively cohomology). Hence, from now on, we will denote indifferently those groups $H_{p}\left(M^{\left(n\right)},\Z\right)\cong F_{p}\oplus T_{p}$ (respectively $H^{p}\left(M^{\left(n\right)},\Z\right)\cong F^{p}\oplus T^{p}$.) 

Two other isomorphisms will be very useful in the following:
\begin{item}
\item[-] $T^{p+1}\cong T_{n-p-1} = T_{q}$ (Poincar{\'e} Duality),
\item[-] $T^{p+1}\cong T_{p}$ (Universal Coefficient Theorem).
\end{item}

Later on, we will prefer to write
\begin{equation}
\mathcal{Z}^{p}_{\mathrm{BF}_{k}}\left(M^{\left(n\right)}\right)
= \abs{T_{p}}\abs{\Hom{T_{p}}{\slfrac{\Z}{k\Z}}}.
\end{equation}
\end{remark}

\section{Generalized Abelian TV invariant}

% \subsection{Setting}
% 
% In the following,
% \begin{itemize}
% \item[-] $C_{p}$ is a free Abelian group generated by the $p$-cells,
% \item[-] $\oppartial_{p}:C_{p}\to C_{p-1}$ is the boundary operator of the chain complex $C_{\bullet}$,
% \item[-] $H_{p}\left(M,\Z\right) = \slfrac{\noyau\left(\partial_{p}\right)}{\image\left(\partial_{p+1}\right)} = \slfrac{Z_{p}}{B_{p}}$ is the homology group of degree $p$ of the chain complex $C_{\bullet}$,
% \item[-] $C^{p}_{k} = \left\lbrace l^{\left(p\right)}_{k}: C_{p}\to\slfrac{\Z}{k\Z}\right\rbrace$ is the set of $\left(k,p\right)$-labelings of $M$,
% \item[-] $\opd^{p}:C^{p}_{k}\to C^{p+1}_{k}$ is the coboundary operator of the cochain complex $C^{\bullet}_{k}$ defined by 
% \begin{equation}
% \left(\opd^{p}l^{\left(p\right)}_{k}\right)\left(e^{a}_{p+1}\right) = l^{\left(p\right)}_{k}\left(\oppartial_{p+1}e^{a}_{p+1}\right)
% \end{equation}
% \item[-] $H^{p}_{k}\left(M,\slfrac{\Z}{k\Z}\right) = \slfrac{\noyau\left(\opd^{p}\right)}{\image\left(\opd^{p-1}\right)} = \slfrac{Z^{p}_{k}}{B^{p}_{k}}$ is the cohomology group of degree $p$ of the chain complex $C^{\bullet}_{k}$ with value in $\slfrac{\Z}{k\Z}$,
% \end{itemize}

\subsection{The original $\mathcal{U}_{q}\!\left(\mathfrak{sl}_{2}\!\left(\mathbb{C}\right)\right)$ 3D case \cite{TV1992}}

In this subsection, $M^{\left(3\right)}$ is a closed connected oriented smooth $3$-manifold.

In this subsection, and in this subsection only, $q$ is a root of unity. Unlike $\mathfrak{sl}_{2}\!\left(\mathbb{C}\right)$, the set of isomorphism classes of representations of $\mathcal{U}_{q}\!\left(\mathfrak{sl}_{2}\!\left(\mathbb{C}\right)\right)$, quantum deformation of parameter $q$ of the enveloping algebra of $\mathfrak{sl}_{2}\!\left(\mathbb{C}\right)$, is finite. Through misuse of language, we will identify the set of isomorphism classes of representations of $\mathcal{U}_{q}\!\left(\mathfrak{sl}_{2}\!\left(\mathbb{C}\right)\right)$ with its (finite) indexing set $I_{q}$. 

Introduce the quantities $w_{i} = \mathrm{dim}\left(i\right)$ and $w^{2} = \sum\limits_{i\in I}w^{2}_{i}$ that are necessary for the construction of the invariant below. The exact expressions of $I_{q}$, $w_{i}$ and $w^{2}$ as a function of $q$ are specified in \cite{TV1992} but are irrelevant for what we want to show here.

Consider now a triangulation of $M^{\left(3\right)}$, \textit{i.e.} the collection of a (finite) set of vertices $C_{0}$, a (finite) set of edges $C_{1}$ connecting those vertices, a (finite) set of triangular faces $C_{2}$ bounded by those edges, and a (finite) set of tetrahedra $C_{3}$ bounded by those faces.

Define the set of $p$-labelings to be $C^{p}_{q} = \left\lbrace l^{\left(p\right)}:C_{p}\longrightarrow I_{q}\right\rbrace$. We will be mostly interested here in the case $p=1$ (labeling of edges). Given a labeling $l^{\left(1\right)}\in C^{1}_{q}$, each tetrahedron $e_{\left(3\right)}\in C_{3}$ becomes a labeled tetrahedron $e^{l^{\left(1\right)}}_{3}$, to which we associate a number $\abs{e_{\left(3\right)}^{l^{\left(1\right)}}}\in\mathbb{C}$, which is here a quantum $6j$-symbol, see Figure \ref{Tetraedre_Uqsl2}.

\begin{figure}
\begin{center}
\includegraphics[scale = 0.75]{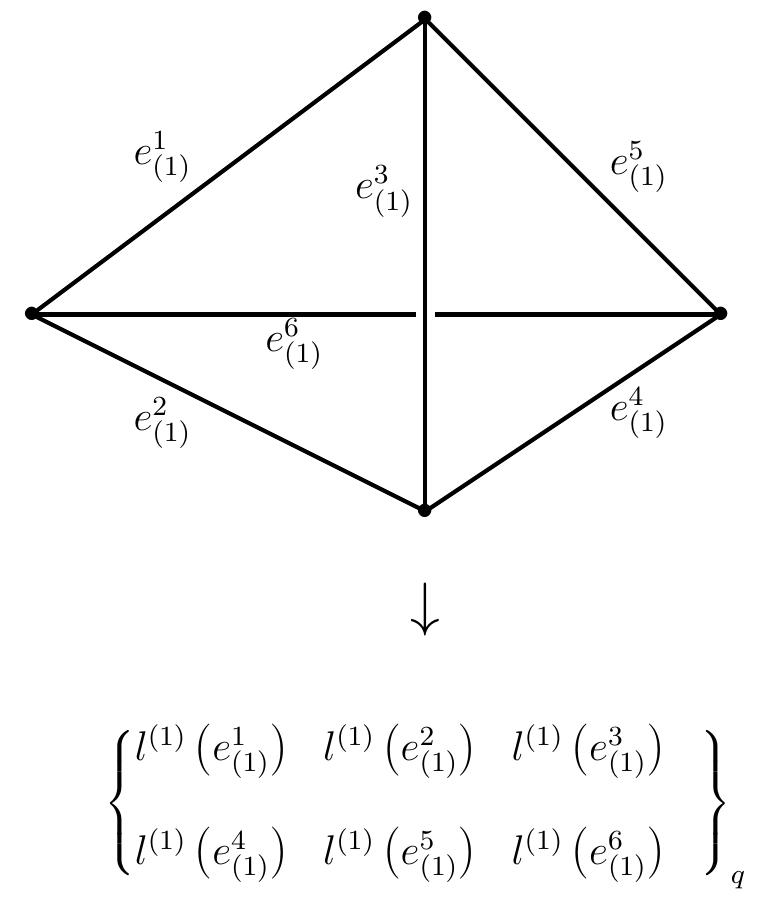}
\caption[]{Quantum $6j$-symbol associated with a tetrahedron in the $\mathcal{U}_{q}\!\left(\mathfrak{sl}_{2}\!\left(\mathbb{C}\right)\right)$ case.}
\label{Tetraedre_Uqsl2}
\end{center}
\end{figure}

\begin{remark}
For a closed manifold $M^{\left(3\right)}$, any triangular face belongs to two tetrahedra. If, from the viewpoint of one tetrahedron, a given face is labeled by $i$, $j$ and $k$, then from the viewpoint of the other tetrahedron to which belongs this face, it is labeled by $i^{*}$, $j^{*}$ and $k^{*}$, where $i^{*}$ means the dual representation of the $i$-th representation.
\end{remark}

Finally, we have:
\begin{theorem}[V. Turaev, O. Viro \cite{TV1992}]
The quantity
\begin{equation}
\label{Z_TV_Uqsl2}
\mathcal{Z}_{\mathrm{TV}_{q}}\left(M^{\left(3\right)}\right)
=\frac{1}{w^{2\abs{C_{0}}}}\sum\limits_{l^{\left(1\right)}\in C^{1}_{q}}
\left(\left(\prod\limits_{e_{\left(1\right)}\in C_{1}}w_{l^{\left(1\right)}\left(e_{\left(1\right)}\right)}\right)
\left(\prod\limits_{e_{\left(3\right)}\in C_{3}}\abs{e^{l^{\left(1\right)}}_{\left(3\right)}}\right)\right)\in\mathbb{C}
\end{equation}
is independent of the triangulation and defines thus an invariant of $M^{\left(3\right)}$.
\end{theorem}

Barrett and Westbury \cite{BW1996} extended the construction to spherical categories, \textit{i.e.} monoidal categories with a notion of duality. Balsam and Kirillov \cite{BK2010} extended in their turn the construction to $3$-manifolds provided with a polyhedral decomposition, \textit{i.e.} a cellular decomposition obtained from a triangulation by removing vertices, edges and faces.

\subsection{The 3D $\slfrac{\Z}{k\Z}$ case}

In this subsection, $M^{\left(3\right)}$ is still a closed connected oriented smooth $3$-manifold provided with a triangulation. For further convenience, let us write this triangulation as a chain complex:
\begin{equation}
\label{chain_complex_cells}
C_{3}\overset{\oppartial_{3}}{\longrightarrow}
C_{2}\overset{\oppartial_{2}}{\longrightarrow}
C_{1}\overset{\oppartial_{1}}{\longrightarrow}
C_{0}
\end{equation}
where $\oppartial_{i}$ is the obvious boundary operator.

The construction in the previous section relies on the fact that $I_{q}$ is finite, which is made possible thanks to the quantum deformation of  the enveloping algebra of $\mathfrak{sl}_{2}\!\left(\mathbb{C}\right)$. But if we want an Abelian version of this construction, we realize immediately that this quantum deformation trick does not work for $\mathfrak{u}\!\left(1\right)$. However, there is at least one other way to extract from $\mathrm{U}\!\left(1\right)$ a finite set of representations with a ``cutting parameter'' $k$: We may consider $\slfrac{\Z}{k\Z}$, which is isomorphic to a finite subgroup of $\mathrm{U}\!\left(1\right)$. 

We will not derive the details from the theory of spherical categories here, but we will refer to \cite{MT2016JMP,MT2016NPB} and recall that
\begin{itemize}
\item[-] $I_{q}$ in the $\mathcal{U}_{q}\!\left(\mathfrak{sl}_{2}\!\left(\mathbb{C}\right)\right)$ case is replaced by $I_{k} = \left\lbrace 0,\hdots,\left(k-1\right)\right\rbrace$ in the $\slfrac{\Z}{k\Z}$ case,
\item[-] to the chain complex \eqref{chain_complex_cells}, we can associate the cochain complex of the $p$-labelings $C^{p}_{k} = \left\lbrace l^{\left(p\right)}:C_{p}\longrightarrow I_{k}\right\rbrace$:
\begin{equation}
\label{cochain_complex_cells}
C_{0}\overset{\opd_{0}}{\longrightarrow}
C_{1}\overset{\opd_{1}}{\longrightarrow}
C_{2}\overset{\opd_{2}}{\longrightarrow}
C_{3}
\end{equation}
where $\opd$ is defined by duality as
\begin{equation}
\left(\opd^{p}l^{\left(p\right)}\right)\left(e_{\left(p+1\right)}\right):= l^{\left(p\right)}\left(\oppartial_{p+1}e_{\left(p+1\right)}\right)
\end{equation}
for $l^{\left(p\right)}\in C^{p}_{k}$ and $e_{\left(p+1\right)}\in C_{p+1}$. Hence, in particular, $\left(\opd^{1} l^{\left(1\right)}\right)$ becomes a labeling of the faces,
\item[-] the quantum $6j$-symbols associated with a tetrahedron in the $\mathcal{U}_{q}\!\left(\mathfrak{sl}_{2}\!\left(\mathbb{C}\right)\right)$ case are replaced here by a product of Kronecker symbols modulo $k$\footnote{What we mean by ``Kronecker symbols modulo $k$'' is $\opdelta^{\left[k\right]}_{a} = \left\lbrace\begin{tabular}{ll} $1$ & if $a\equiv 0\left[k\right]$ \\ $0$ & otherwise\end{tabular}\right.$.}, each Kronecker symbol representing the colored edges of a face of the tetrahedron considered:
\begin{align}
\nonumber
\abs{e^{l^{\left(1\right)}}_{\left(3\right)}} =&
\opdelta^{\left[k\right]}_{l^{\left(1\right)}\left(e^{1}_{\left(1\right)}\right)+l^{\left(1\right)}\left(e^{2}_{\left(1\right)}\right)-l^{\left(1\right)}\left(e^{3}_{\left(1\right)}\right)}
\opdelta^{\left[k\right]}_{l^{\left(1\right)}\left(e^{3}_{\left(1\right)}\right)+l^{\left(1\right)}\left(e^{4}_{\left(1\right)}\right)-l^{\left(1\right)}\left(e^{5}_{\left(1\right)}\right)}\\
&\opdelta^{\left[k\right]}_{l^{\left(1\right)}\left(e^{5}_{\left(1\right)}\right)-l^{\left(1\right)}\left(e^{6}_{\left(1\right)}\right)-l^{\left(1\right)}\left(e^{1}_{\left(1\right)}\right)}
\opdelta^{\left[k\right]}_{l^{\left(1\right)}\left(e^{6}_{\left(1\right)}\right)-l^{\left(1\right)}\left(e^{4}_{\left(1\right)}\right)-l^{\left(1\right)}\left(e^{2}_{\left(1\right)}\right)}
\end{align}
or, in a more condensed notation:
\begin{equation}
\abs{e_{\left(3\right)}^{l^{\left(1\right)}}} 
= \prod\limits_{e_{\left(2\right)}\in C_{2}\,\vert\,e_{\left(2\right)}\subset\oppartial_{3}e_{\left(3\right)}}
\opdelta^{\left[k\right]}_{l^{\left(1\right)}\left(\oppartial_{2}e_{\left(2\right)}\right)},
\end{equation}
see Figure \ref{Tetraedre_ZkZ}.
\end{itemize}

\begin{figure}
\begin{center}
\includegraphics[scale = 0.75]{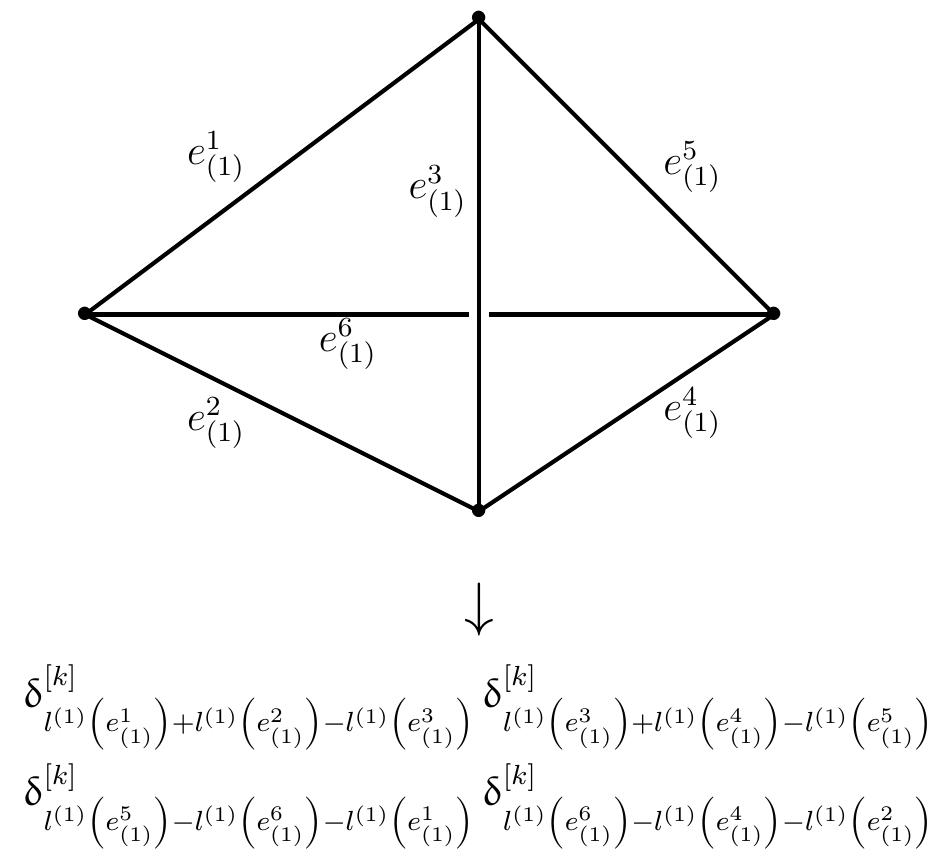}
\caption[]{Kronecker symbols modulo $k$ associated with a tetrahedron in the $\slfrac{\Z}{k\Z}$ case.}
\label{Tetraedre_ZkZ}
\end{center}
\end{figure}

Hence, knowing that for all $i\in I_{k}$, $w_{i} = 1$, we can adapt the formula of the TV invariant \eqref{Z_TV_Uqsl2} to the Abelian case\footnote{The superscript $1$ in $\mathcal{Z}^{1}_{\mathrm{TV}_{k}}\left(M^{\left(3\right)}\right)$ will be explained in the next section.}:
\begin{definition}
We set
\begin{align}
\nonumber
\mathcal{Z}^{1}_{\mathrm{TV}_{k}}\left(M^{\left(3\right)}\right)
=& \frac{1}{k^{\abs{C_{0}}-1}}\sum\limits_{l^{\left(1\right)}\in C^{1}_{k}}
\prod\limits_{e_{\left(3\right)}\in C_{3}}\abs{e_{\left(3\right)}^{l^{\left(1\right)}}}\\
=& \frac{1}{k^{\abs{C_{0}}-1}}\sum\limits_{l^{\left(1\right)}\in C^{1}_{k}}
\prod\limits_{e_{\left(3\right)}\in C_{3}}
\prod\limits_{e_{\left(2\right)}\in C_{2}\,\vert\,e_{\left(2\right)}\subset\oppartial_{3}e_{\left(3\right)}}
\opdelta^{\left[k\right]}_{l^{\left(1\right)}\left(\oppartial_{2}e_{\left(2\right)}\right)}\in\mathbb{C}
\end{align}

Note that the gluing condition of two faces along a given edge translates algebraically into the fact that this edge has to carry a representation $i$ from the viewpoint of one face, and $i^{*}$ from the viewpoint of the other face. Here, the duality is simply $i^{*} \equiv \left(k-i\right) \,\left[k\right]$ or $i^{*} \equiv-i \,\left[k\right]$ (strictly speaking, in our settings, $-i\notin I_{k}$ but $\left(k-i\right)\in I_{k}$.) So in the definition above, each face appears twice, with one orientation and the opposite one. But the Kronecker symbol does not make a difference between the two orientations, so we may simply write:
\begin{equation}
\label{Def_TV_ab_3D}
\mathcal{Z}^{1}_{\mathrm{TV}_{k}}\left(M^{\left(3\right)}\right)
= \frac{1}{k^{\abs{C_{0}}-1}}\sum\limits_{l^{\left(1\right)}\in C^{1}_{k}}
\prod\limits_{e_{\left(2\right)}\in C_{2}}\opdelta^{\left[k\right]}_{l^{\left(1\right)}\left(\oppartial_{2} e_{\left(2\right)}\right)}\in\mathbb{C}
\end{equation}
\end{definition}

\begin{remark}
We realize that the assumption that we have \textit{stricto sensu} a triangulation of $M^{\left(3\right)}$ is not necessary and that a cellular decomposition is sufficient, although to make sense of formula \eqref{Def_TV_ab_3D}, $C_{2}$ cannot be empty (but we will come back to this point at the end of this section.).
\end{remark}

\begin{remark}
Anticipating one of the next result, it is a good point to recall an important fact. The Universal Coefficient Theorem \cite{BT1982} claims that, for any abelian group $G$,
\begin{equation}
H^{p}\left(M^{\left(n\right)},G\right)\simeq\Hom{H_{p}\left(M^{\left(n\right)},\Z\right)}{G}\oplus\Ext{H_{p-1}\left(M^{\left(n\right)},\Z\right)}{G},
\end{equation}
\textit{i.e.} in general, the cohomology with coefficients in $G$ is \textit{not} the group of $G$-valued homomorphisms over the homology (with coefficients in $\Z$). However, if $G$ is divisible, then $\Ext{H_{p-1}\left(M^{\left(n\right)},\Z\right)}{G}$ is trivial and as a consequence, \textit{in this case}, the cohomology with coefficients in $G$ \textit{is} the group of $G$-valued homomorphisms over the homology (with coefficients in $\Z$), \textit{i.e.}
\begin{equation}
H^{p}\left(M^{\left(n\right)},G\right)\cong\Hom{H_{p}\left(M^{\left(n\right)},\Z\right)}{G}.
\end{equation}
Moreover, the group $\slfrac{\Z}{k\Z}$ is a divisible group. Hence, the group of cohomology modulo $k$ of degree $p$ is
\begin{equation}
H^{p}_{k} := H^{p}\left(M^{\left(n\right)},\slfrac{\Z}{k\Z}\right)\simeq\Hom{H_{p}\left(M^{\left(n\right)},\Z\right)}{\slfrac{\Z}{k\Z}}.
\end{equation}
\end{remark}

\begin{claim}
\label{Z_TV_ab_3D_H_k}
An equivalent definition of $\mathcal{Z}^{1}_{\mathrm{TV}_{k}}\left(M^{\left(3\right)}\right)$ is
\begin{equation}
\label{Z_TV_ab_3D}
\mathcal{Z}^{1}_{\mathrm{TV}_{k}}\left(M^{\left(3\right)}\right) 
= \frac{1}{k^{\abs{C_{0}}-1}}\sum\limits_{l^{\left(1\right)}\in C^{1}_{k}}
\prod\limits_{e_{\left(2\right)}\in C_{2}}\opdelta^{\left[k\right]}_{\left(\opd^{1}l^{\left(1\right)}\right)\left(e_{\left(2\right)}\right)}\in\mathbb{C}
\end{equation}
from which we can show that
\begin{equation}
\label{Z_TV_ab_3D}
\mathcal{Z}^{1}_{\mathrm{TV}_{k}}\left(M^{\left(3\right)}\right) 
= k\frac{\abs{H^{1}_{k}}}{\abs{H^{0}_{k}}}
\end{equation}
\end{claim}

\begin{proof}
First, denote $Z^{p}_{k} = \noyau{\opd^{p}}$ the group of $p$-cocycles modulo $k$, $B^{p}_{k} = \image{\opd^{p-1}}$ the group of $p$-coboundaries modulo $k$. By definition, $H^{p}_{k} = \slfrac{Z^{p}_{k}}{B^{p}_{k}}$. Then, 
\begin{equation}
\abs{H^{1}_{k}} = \frac{\abs{Z^{1}_{k}}}{\abs{B^{1}_{k}}}
\end{equation}
and
\begin{equation}
\abs{B^{1}_{k}} = \frac{\abs{C^{0}_{k}}}{\abs{Z^{0}_{k}}}.
\end{equation}

Now, observe equation \eqref{Z_TV_ab_3D_H_k}. We realize that, apart from the normalization, we are actually counting the number $\abs{Z^{1}_{k}}$ of $1$-cocycles modulo $k$ for the cochain complex \eqref{cochain_complex_cells}. This is of course not an invariant of $M^{\left(3\right)}$, but
\begin{align}
\mathcal{Z}^{1}_{\mathrm{TV}_{k}}\left(M^{\left(3\right)}\right) 
=& \frac{1}{k^{\abs{C_{0}}-1}}\abs{Z^{1}_{k}}\\
=& \frac{1}{k^{\abs{C_{0}}-1}}\abs{H^{1}_{k}}\cdot\abs{B^{1}_{k}}\\
=& \frac{1}{k^{\abs{C_{0}}-1}}\abs{H^{1}_{k}}\cdot\frac{\abs{C^{0}_{k}}}{\abs{Z^{0}_{k}}}\\
=& \frac{1}{k^{\abs{C_{0}}-1}}\abs{H^{1}_{k}}\cdot\frac{k^{\abs{C_{0}}}}{\abs{H^{0}_{k}}}\\
\mathcal{Z}^{1}_{\mathrm{TV}_{k}}\left(M^{\left(3\right)}\right)
=& k\frac{\abs{H^{1}_{k}}}{\abs{H^{0}_{k}}},
\end{align}
whence the result.
\end{proof}

\begin{consequence}
If $M^{\left(3\right)}$ has only one connected component, as we assume here, \textit{i.e.} if $\abs{H^{0}_{k}} = k$, then
\begin{align}
\label{TV_Cohom}
\mathcal{Z}^{1}_{\mathrm{TV}_{k}}\left(M^{\left(3\right)}\right)
=& \abs{H^{1}_{k}}
\end{align}
\end{consequence}

\begin{remark}
Taking into account equation \eqref{TV_Cohom}, if $C_{2} = \varnothing$, it is reasonable to set 
\begin{equation}
\label{No_Cell}
\mathcal{Z}^{1}_{\mathrm{TV}_{k}}\left(M^{\left(3\right)}\right) = 1,
\end{equation}
as the absence of $2$-cells means there is no homology in degree $2$ and by Poincar{\'e} duality, there is no cohomology in degree $1$.
\end{remark}

\begin{remark}
We verify \textit{a posteriori} that what we are computing is indeed totally independent from the choice of the cellular decomposition.
\end{remark}

\begin{figure}[h]
\begin{center}
\includegraphics[scale=1.]{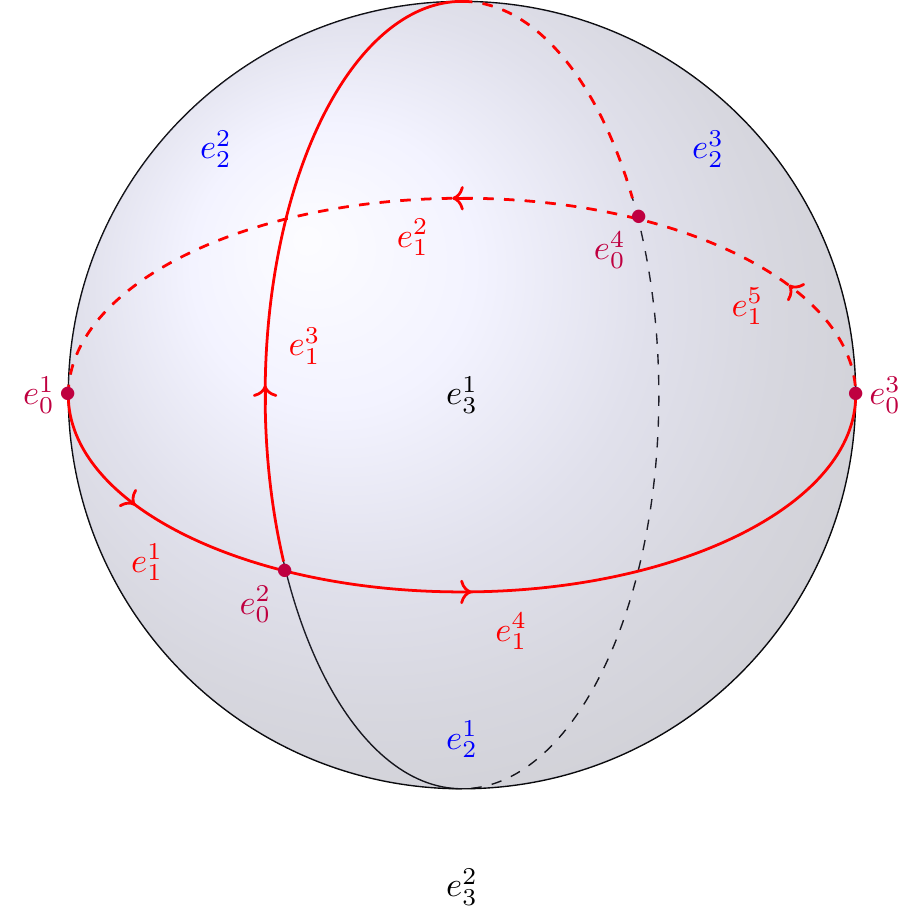}
\caption{A particular cellular decomposition of $S^{3}$}
\label{Decompo_S3}
\end{center}
\end{figure}

\begin{example}
Consider the decomposition of $S^{3}$ given by Figure \ref{Decompo_S3}. Then,
\begin{align*}
\mathcal{Z}^{1}_{\mathrm{TV}_{k}}\left(S^{3}\right)
&= \frac{1}{k^{\abs{C_{0}}-1}}
\sum\limits_{l^{\left(1\right)}\in C^{1}_{k}}\prod\limits_{a=1}^{3}
\opdelta^{\left[k\right]}_{\left(\opd^{1}l^{\left(1\right)}\right)\left(e^{a}_{\left(2\right)}\right)}\\
&= \frac{1}{k^{4-1}}
\sum\limits_{l^{\left(1\right)}\in C^{1}_{k}}
\opdelta^{\left[k\right]}_{l^{\left(1\right)}\left(\oppartial_{2}e^{1}_{\left(2\right)}\right)}
\opdelta^{\left[k\right]}_{l^{\left(1\right)}\left(\oppartial_{2}e^{2}_{\left(2\right)}\right)}
\opdelta^{\left[k\right]}_{l^{\left(1\right)}\left(\oppartial_{2}e^{3}_{\left(2\right)}\right)}\\
&= \frac{1}{k^{3}}
\sum\limits_{l^{\left(1\right)}\in C^{1}_{k}}
\opdelta^{\left[k\right]}_{l^{\left(1\right)}\left(e^{1}_{\left(1\right)}+e^{2}_{\left(1\right)}
+e^{4}_{\left(1\right)}+e^{5}_{\left(1\right)}\right)}
\opdelta^{\left[k\right]}_{l^{\left(1\right)}\left(e^{1}_{\left(1\right)}+e^{2}_{\left(1\right)}
+e^{3}_{\left(1\right)}\right)}
\opdelta^{\left[k\right]}_{l^{\left(1\right)}\left(e^{3}_{\left(1\right)}-e^{4}_{\left(1\right)}
-e^{5}_{\left(1\right)}\right)}\\
&= \frac{1}{k^{3}}
\sum\limits_{l^{\left(1\right)}\in C^{1}_{k}}
\opdelta^{\left[k\right]}_{l^{\left(1\right)}_{k}\left(e^{1}_{\left(1\right)}\right)
+l^{\left(1\right)}\left(e^{2}_{\left(1\right)}\right)
+l^{\left(1\right)}\left(e^{4}_{\left(1\right)}\right)
+l^{\left(1\right)}\left(e^{5}_{\left(1\right)}\right)}\\
&\hspace{2.cm}\opdelta^{\left[k\right]}_{l^{\left(1\right)}\left(e^{1}_{\left(1\right)}\right)
+l^{\left(1\right)}\left(e^{2}_{\left(1\right)}\right)
+l^{\left(1\right)}\left(e^{3}_{\left(1\right)}\right)}
\opdelta^{\left[k\right]}_{l^{\left(1\right)}\left(e^{3}_{\left(1\right)}\right)
-l^{\left(1\right)}\left(e^{4}_{\left(1\right)}\right)
-l^{\left(1\right)}\left(e^{5}_{\left(1\right)}\right)}\\
&= \frac{1}{k^{3}}
\sum\limits_{i,j,l,m,n=0}^{k-1}
\opdelta^{\left[k\right]}_{i+j+m+n}
\opdelta^{\left[k\right]}_{i+j+l}
\opdelta^{\left[k\right]}_{l-m-n}\\
&= \frac{1}{k^{3}}
\sum\limits_{i,j,m,n=0}^{k-1}
\opdelta^{\left[k\right]}_{i+j+m+n}\\
&= \frac{1}{k^{3}}
\sum\limits_{i,j,m=0}^{k-1}1\\
&= \frac{1}{k^{3}}k^{3}\\
\mathcal{Z}^{1}_{\mathrm{TV}_{k}}\left(S^{3}\right)
&= 1
\end{align*} 

Note that this result is trivial from \eqref{No_Cell} if we consider $S^{3} = B^{3}\underset{\oppartial B^{3}}{\cup}\mathrm{pt}$, which is a decomposition that has only one $0$-cell and one $3$-cell, and no $1$-cell nor $2$-cell.
\end{example}

\begin{remark}
If we had followed the standard convention of normalization, \textit{i.e.} if we had defined $\mathcal{Z}^{1}_{\mathrm{TV}_{k}}\left(M^{\left(3\right)}\right)$ as
\begin{equation}
\mathcal{Z}^{1}_{\mathrm{TV}_{k}}\left(M^{\left(3\right)}\right)
= \frac{1}{k^{\abs{C_{0}}}}\sum\limits_{l^{\left(1\right)}\in C^{1}_{k}}
\left(\prod\limits_{e_{\left(3\right)}\in C_{3}}\opdelta^{\left[k\right]}_{l^{\left(1\right)}\left(\oppartial_{2}e_{\left(2\right)}\right)}\right)\in\mathbb{C}
\end{equation}
then we would have had
\begin{equation}
\mathcal{Z}^{1}_{\mathrm{TV}_{k}}\left(S^{3}\right) = \frac{1}{k}
\end{equation}
and 
\begin{equation}
\mathcal{Z}^{1}_{\mathrm{TV}_{k}}\left(S^{1}\times S^{2}\right) = 1.
\end{equation}
This definition is the appropriate one in order for $\mathcal{Z}^{1}_{\mathrm{TV}_{k}}\left(M^{\left(3\right)}\right)$ to define a TQFT satisfying Atiyah-Segal axioms \cite{ATI1988,SEG1988,SEG2001}. In this framework, the invariant is also defined for compact $3$-manifolds with boundary, which is a case we are not considering here. We focus here on \textit{closed} $3$-manifolds, and we shifted the number of vertices by $1$ in the definition of $\mathcal{Z}^{1}_{\mathrm{TV}_{k}}\left(M^{\left(3\right)}\right)$ in order to have
\begin{equation}
\mathcal{Z}^{1}_{\mathrm{TV}_{k}}\left(S^{3}\right) = 1
\end{equation}
for later convenience. 
\end{remark}

\begin{remark}
This TV construction (in 3D with a triangulation) may remind the reader of the Dijkgraaf-Witten (DW) construction \cite{DW1990}. According to \cite{TV2017}, for a pair $\left(G,\alpha\right)$ where $G$ is a finite group (not necessarily Abelian in general) and  $\alpha$ is a $3$-cocycle of the classifying space $BG$ of $G$, then there exists a spherical fusion category $\mathcal{C}$ such that the TV construction for $\mathcal{C}$ matches with the DW construction for the pair $\left(G,\alpha\right)$:
\begin{equation}
\mathcal{Z}_{\mathrm{TV}_{\mathcal{C}}}\left(M^{\left(3\right)}\right)
 = \mathcal{Z}_{\mathrm{DW}_{\left(G,\alpha\right)}}\left(M^{\left(3\right)}\right)
\end{equation}
The spherical fusion category $\mathcal{C}$ is explicitly built from $G$, $\alpha$ playing the role of associator.

The question here is somewhat the opposite: If we properly follow the categorical construction that leads to \eqref{Def_TV_ab_3D}, are we able to find a pair $\left(G,\alpha\right)$ such that 
\begin{equation}
\mathcal{Z}^{1}_{\mathrm{TV}_{k}}\left(M^{\left(3\right)}\right)
 = \mathcal{Z}_{\mathrm{DW}_{\left(G,\alpha\right)}}\left(M^{\left(3\right)}\right)?
\end{equation}
The answer is positive, with $G = \slfrac{\Z}{k\Z}$ and $\alpha$ trivial, in the sense that the associativity of the underlying category is not twisted. 

Remark that the Kronecker symbols modulo $k$ of formula \eqref{Def_TV_ab_3D} do not appear explicitly in the DW construction. However, this construction consists of summing only over the labels satisfying the condition that, if three labels are associated with the edges of a triangle of the triangulation of $M$, then their sum must be zero (in additive notations.) If, instead, we want to sum over all the labels, then for $G = \slfrac{\Z}{k\Z}$, this condition translates into Kronecker symbols modulo $k$ associated with each triangle of the triangulation of $M$.
\end{remark}

\subsection{The nD $\slfrac{\Z}{k\Z}$ case}

In this subsection, $M^{\left(n\right)}$ is a closed connected oriented smooth $n$-manifold.

Clearly, formula \eqref{Z_TV_ab_3D} can be generalized to manifolds of any dimension $n$, and we can decide to label the cells of any dimension $p$. 
\begin{definition}
For all $p\in\left\lbrace 0,\hdots,\left(n-1\right)\right\rbrace$, if $C_{p+1}\neq\varnothing$ then
\begin{equation}
\label{Def_TV_ab_nD}
\mathcal{Z}^{p}_{\mathrm{TV}_{k}}\left(M^{\left(n\right)}\right)
:= \frac{1}{k^{\abs{C_{p-1}} - \abs{C_{p-2}} + \hdots - \varepsilon_{p}\abs{C_{0}}+\varepsilon_{p}}}
\sum\limits_{l^{\left(p\right)}_{k}\in C^{p}_{k}}\prod\limits_{a=1}^{\abs{C_{p+1}}}
\opdelta^{\left[k\right]}_{\left(\opd^{p}l^{\left(p\right)}_{k}\right)\left(e^{a}_{p+1}\right)},
\end{equation}
where $\varepsilon_{p} = \left(-1\right)^{p}$, and
\begin{equation}
\mathcal{Z}^{p}_{\mathrm{TV}_{k}}\left(M^{\left(n\right)}\right) = 1
\end{equation}
otherwise.
\end{definition}

\begin{remark}
The normalization in formula \eqref{Def_TV_ab_nD} was chosen in such a way that the quantity $\mathcal{Z}^{p}_{\mathrm{TV}_{k}}\left(M^{\left(n\right)}\right)$ is truly invariant and satisfies $\mathcal{Z}^{p}_{\mathrm{TV}_{k}}\left(S^{n}\right) = 1$ for all $p\in\left\lbrace 0,\hdots,\left(n-1\right)\right\rbrace$.
\end{remark}

\begin{proposition}
We can write that
\begin{align*}
&\mathcal{Z}^{p}_{\mathrm{TV}_{k}}\left(M^{\left(n\right)}\right)=\\
&\begin{cases}
\abs{\Hom{H_{p}}{\slfrac{\Z}{k\Z}}}\dfrac{\abs{\Hom{F_{p-2}}{\slfrac{\Z}{k\Z}}}\hdots\abs{\Hom{F_{0}}{\slfrac{\Z}{k\Z}}}}{\abs{\Hom{F_{p-1}}{\slfrac{\Z}{k\Z}}}\hdots\abs{\Hom{F_{1}}{\slfrac{\Z}{k\Z}}}}\dfrac{1}{\abs{\slfrac{\Z}{k\Z}}}
\,\mbox{if $p$ is even}\\
\abs{\Hom{H_{p}}{\slfrac{\Z}{k\Z}}}\dfrac{\abs{\Hom{F_{p-2}}{\slfrac{\Z}{k\Z}}}\hdots\abs{\Hom{F_{1}}{\slfrac{\Z}{k\Z}}}}{\abs{\Hom{F_{p-1}}{\slfrac{\Z}{k\Z}}}\hdots\abs{\Hom{F_{0}}{\slfrac{\Z}{k\Z}}}}\abs{\slfrac{\Z}{k\Z}}
\,\mbox{if $p$ is odd}
\end{cases}
\end{align*}
or
\begin{align*}
&\mathcal{Z}^{p}_{\mathrm{TV}_{k}}\left(M^{\left(n\right)}\right)=\\
&\begin{cases}
\abs{\Hom{T_{p}}{\slfrac{\Z}{k\Z}}}\dfrac{\abs{\Hom{F_{p}}{\slfrac{\Z}{k\Z}}}\hdots\abs{\Hom{F_{0}}{\slfrac{\Z}{k\Z}}}}{\abs{\Hom{F_{p-1}}{\slfrac{\Z}{k\Z}}}\hdots\abs{\Hom{F_{1}}{\slfrac{\Z}{k\Z}}}}\dfrac{1}{\abs{\slfrac{\Z}{k\Z}}}
\,\mbox{if $p$ is even}\\
\abs{\Hom{T_{p}}{\slfrac{\Z}{k\Z}}}\dfrac{\abs{\Hom{F_{p}}{\slfrac{\Z}{k\Z}}}\hdots\abs{\Hom{F_{1}}{\slfrac{\Z}{k\Z}}}}{\abs{\Hom{F_{p-1}}{\slfrac{\Z}{k\Z}}}\hdots\abs{\Hom{F_{0}}{\slfrac{\Z}{k\Z}}}}\abs{\slfrac{\Z}{k\Z}}
\,\mbox{if $p$ is odd}
\end{cases}
\end{align*}
\end{proposition}

\begin{proof}
The point of this computation is to iterate the process of the proof of Claim \ref{Z_TV_ab_3D_H_k}.

Indeed, from definition \eqref{Def_TV_ab_nD}, we understand that we are counting the number $\abs{Z^{p}_{k}}$ of $p$-cocycles modulo $k$. But
\begin{equation}
\abs{Z^{p}_{k}} = \abs{H^{p}_{k}}\abs{B^{p}_{k}} = \frac{\abs{C^{p-1}_{k}}}{\abs{Z^{p-1}_{N}}}
\end{equation}
whence
\begin{equation}
\label{Zpk_non_simplif}
\abs{Z^{p}_{k}} = 
\begin{cases}
	k^{\abs{C_{p-1}} - \abs{C_{p-2}} + \hdots - \abs{C_{0}}} 
	\frac{\abs{H^{p}_{k}}\abs{H^{p-2}_{k}}\hdots\abs{H^{0}_{k}}}
	{\abs{H^{p-1}_{k}}\abs{H^{p-3}_{k}}\hdots\abs{H^{1}_{k}}}\,\mbox{if $p$ is even,}\\
	\\
	k^{\abs{C_{p-1}} - \abs{C_{p-2}} + \hdots + \abs{C_{0}}} 
	\frac{\abs{H^{p}_{k}}\abs{H^{p-2}_{k}}\hdots\abs{H^{1}_{k}}}
	{\abs{H^{p-1}_{k}}\abs{H^{p-3}_{k}}\hdots\abs{H^{0}_{k}}}\,\mbox{if $p$ is odd.}	
\end{cases}
\end{equation}

But from the Universal Coefficient Theorem \cite{BT1982}, we can write
\begin{equation}
H^{\bullet}_{k}\cong\Hom{H_{\bullet}}{\slfrac{\Z}{k\Z}}\oplus\Ext{T_{\bullet-1}}{\slfrac{\Z}{k\Z}},
\end{equation}
and moreover
\begin{equation}
\Ext{\slfrac{\Z}{\zeta\Z}}{\slfrac{\Z}{k\Z}}
\cong\slfrac{\Z}{\gcd\left(\zeta,k\right)\Z} 
\cong\Hom{\slfrac{\Z}{\zeta\Z}}{\slfrac{\Z}{k\Z}}.
\end{equation}
Hence,
\begin{equation}
H^{\bullet}_{k}
\cong\Hom{F_{\bullet}}{\slfrac{\Z}{k\Z}}
\oplus\Hom{T_{\bullet}}{\slfrac{\Z}{k\Z}}
\oplus\Hom{T_{\bullet-1}}{\slfrac{\Z}{k\Z}},
\end{equation}
from which we obtain that
\begin{equation}
\abs{H^{\bullet}_{k}}
=\abs{\Hom{F_{\bullet}}{\slfrac{\Z}{k\Z}}}
\abs{\Hom{T_{\bullet}}{\slfrac{\Z}{k\Z}}}
\abs{\Hom{T_{\bullet-1}}{\slfrac{\Z}{k\Z}}}
\end{equation}
which, introduced in equation \eqref{Zpk_non_simplif}, gives by telescoping simplifications the result expected.
\end{proof}

\begin{example}
Still with the decomposition on Figure \ref{Decompo_S3} of $S^{3}$, we can compute
\begin{align*}
\mathcal{Z}^{0}_{\mathrm{TV}_{k}}\left(S^{3}\right)
&= \frac{1}{k}
\sum\limits_{l^{0}_{k}\in C^{0}_{k}}\prod\limits_{a=1}^{5}
\opdelta^{\left[k\right]}_{\left(\opd^{0}l^{0}_{k}\right)\left(e^{a}_{1}\right)}\\
&= \frac{1}{k}
\sum\limits_{l^{0}_{k}\in C^{0}_{k}}
\opdelta^{\left[k\right]}_{l^{0}_{k}\left(\oppartial_{1}e^{1}_{1}\right)}
\opdelta^{\left[k\right]}_{l^{0}_{k}\left(\oppartial_{1}e^{2}_{1}\right)}\\
&\qquad\opdelta^{\left[k\right]}_{l^{0}_{k}\left(\oppartial_{1}e^{3}_{1}\right)}
\opdelta^{\left[k\right]}_{l^{0}_{k}\left(\oppartial_{1}e^{4}_{1}\right)}
\opdelta^{\left[k\right]}_{l^{0}_{k}\left(\oppartial_{1}e^{5}_{1}\right)}\\
&= \frac{1}{k}
\sum\limits_{l^{0}_{k}\in C^{0}_{k}}
\opdelta^{\left[k\right]}_{l^{0}_{k}\left(e^{2}_{0}-e^{1}_{0}\right)}
\opdelta^{\left[k\right]}_{l^{0}_{k}\left(e^{1}_{0}-e^{4}_{0}\right)}\\
&\qquad\opdelta^{\left[k\right]}_{l^{0}_{k}\left(e^{4}_{0}-e^{2}_{0}\right)}
\opdelta^{\left[k\right]}_{l^{0}_{k}\left(e^{3}_{0}-e^{2}_{0}\right)}
\opdelta^{\left[k\right]}_{l^{0}_{k}\left(e^{4}_{0}-e^{3}_{0}\right)}\\
&= \frac{1}{k}
\sum\limits_{i,j,l,m=0}^{k-1}
\opdelta^{\left[k\right]}_{j-i}
\opdelta^{\left[k\right]}_{i-m}
\opdelta^{\left[k\right]}_{m-j}
\opdelta^{\left[k\right]}_{l-j}
\opdelta^{\left[k\right]}_{m-l}\\
\mathcal{Z}^{0}_{\mathrm{TV}_{k}}\left(S^{3}\right)
&= 1
\end{align*} 
and
\begin{align*}
\mathcal{Z}^{2}_{\mathrm{TV}_{k}}\left(S^{3}\right)
&= \frac{1}{k^{\abs{C_{1}}-\abs{C_{0}}+\varepsilon_{1}}}
\sum\limits_{l^{2}_{k}\in C^{2}_{k}}\prod\limits_{a=1}^{2}
\opdelta^{\left[k\right]}_{\left(\opd^{2}l^{2}_{k}\right)\left(e^{a}_{3}\right)}\\
&= \frac{1}{k^{5-4+1}}
\sum\limits_{l^{2}_{k}\in C^{2}_{k}}
\opdelta^{\left[k\right]}_{l^{2}_{k}\left(\oppartial_{2}e^{1}_{3}\right)}
\opdelta^{\left[k\right]}_{l^{2}_{k}\left(\oppartial_{2}e^{2}_{3}\right)}\\
&= \frac{1}{k^{2}}
\sum\limits_{l^{2}_{k}\in C^{2}_{k}}
\opdelta^{\left[k\right]}_{l^{2}_{k}\left(e^{1}_{2}+e^{2}_{2}+e^{3}_{2}\right)}
\opdelta^{\left[k\right]}_{l^{2}_{k}\left(e^{1}_{2}+e^{2}_{2}+e^{3}_{2}\right)}\\
&= \frac{1}{k^{2}}
\sum\limits_{i,j,l=0}^{k-1}
\opdelta^{\left[k\right]}_{i+j+l}\\
&= \frac{1}{k^{2}}
\sum\limits_{i,j=0}^{k-1}1\\
&= \frac{1}{k^{2}}k^{2}\\
\mathcal{Z}^{2}_{\mathrm{TV}_{k}}\left(S^{3}\right)
&= 1
\end{align*}
\end{example}

\begin{example}
Consider the following cellular decomposition of $S^{n}$: For all $p\in\left\lbrace 0, \hdots, n\right\rbrace,\,e^{1}_{\left(p\right)}, e^{2}_{\left(p\right)} = B^{p}$, and for all $p\in\left\lbrace 1, \hdots, n\right\rbrace$, $\oppartial_{p}e^{1}_{\left(p\right)} = -\oppartial_{p} e^{2}_{\left(p\right)} = S^{p-1} = e^{2}_{\left(p-1\right)} + e^{1}_{\left(p-1\right)}$. Then for all $p\in\left\lbrace 0, \hdots, \left(n-1\right)\right\rbrace$,
\begin{align*}
\mathcal{Z}^{p}_{\mathrm{TV}_{k}}\left(S^{n}\right)
&= \frac{1}{k^{2 - 2 + \hdots - 2\varepsilon_{p}+\varepsilon_{p}}}
\sum\limits_{l^{\left(p\right)}\in C^{p}_{k}}\prod\limits_{a=1}^{2}
\opdelta^{\left[k\right]}_{\left(\opd^{p}l^{\left(p\right)}\right)\left(e^{a}_{\left(p+1\right)}\right)}\\
&= \frac{1}{k}
\sum\limits_{l^{\left(p\right)}\in C^{p}_{k}}
\opdelta^{\left[k\right]}_{l^{\left(p\right)}\left(\oppartial_{p+1}e^{1}_{\left(p+1\right)}\right)}
\opdelta^{\left[k\right]}_{l^{\left(p\right)}\left(\oppartial_{p+1}e^{2}_{\left(p+1\right)}\right)}\\
&= \frac{1}{k}
\sum\limits_{l^{\left(p\right)}\in C^{p}_{k}}
\opdelta^{\left[k\right]}_{l^{\left(p\right)}\left(e^{1}_{\left(p\right)} + e^{2}_{\left(p\right)}\right)}
\opdelta^{\left[k\right]}_{l^{\left(p\right)}\left(- e^{1}_{\left(p\right)} - e^{2}_{\left(p\right)}\right)}\\
&= \frac{1}{k}
\sum\limits_{l^{\left(p\right)}\in C^{p}_{k}}
\opdelta^{\left[k\right]}_{l^{\left(p\right)}\left(e^{1}_{\left(p\right)}\right) 
+ l^{\left(p\right)}\left(e^{2}_{\left(p\right)}\right)}
\opdelta^{\left[k\right]}_{-l^{\left(p\right)}\left(e^{1}_{\left(p\right)}\right) 
- l^{\left(p\right)}\left(e^{2}_{\left(p\right)}\right)}\\
&= \frac{1}{k}
\sum\limits_{i,j=0}^{k-1}
\opdelta^{\left[k\right]}_{i + j}
\opdelta^{\left[k\right]}_{-i - j}\\
&= \frac{1}{k}
\sum\limits_{i,j=0}^{k-1}
\opdelta^{\left[k\right]}_{i + j}\\
&= \frac{1}{k}
\sum\limits_{i=0}^{k-1}1\\
\mathcal{Z}^{p}_{\mathrm{TV}_{k}}\left(S^{n}\right)
&= 1
\end{align*}
\end{example}

% \textcolor{red}{Add other less trivial examples?}

\section{Generalized Turaev-Viro invariant as a discrete BF theory}

\begin{definition}
Given a cellular decomposition
\begin{equation}
C_{n}\overset{\oppartial_{n}}{\longrightarrow}
C_{n-1}\overset{\oppartial_{n-1}}{\longrightarrow}
\hdots\overset{\oppartial_{2}}{\longrightarrow}
C_{1}\overset{\oppartial_{1}}{\longrightarrow}
C_{0}
\end{equation}
of $M^{\left(n\right)}$, we call dual decomposition any cellular decomposition
\begin{equation}
\widetilde{C}_{n}\overset{\oppartial_{n}}{\longrightarrow}
\widetilde{C}_{n-1}\overset{\oppartial_{n-1}}{\longrightarrow}
\hdots\overset{\oppartial_{2}}{\longrightarrow}
\widetilde{C}_{1}\overset{\oppartial_{1}}{\longrightarrow}
\widetilde{C}_{0}
\end{equation}
of $M^{\left(n\right)}$ such that for any $\widetilde{e}^{a}_{\left(p\right)}\in\widetilde{C}_{p}$ and $e^{b}_{\left(n-p\right)}\in C_{n-p}$, 
\begin{equation}
\widetilde{e}^{a}_{\left(p\right)}\cdot e^{b}_{\left(n-p\right)} = \delta^{ab}
\end{equation}
where $\cdot$ is the number of intersections operator. We define also dual labelings from the dual chain complex, exactly the same way we defined the non-dual labelings, \textit{i.e.} $\widetilde{C}^{p}_{k} = \left\lbrace \tilde{l}^{\left(p\right)}:\widetilde{C}_{p}\longrightarrow I_{k}\right\rbrace$.
\end{definition}

\begin{claim}
We can write
\begin{align*}
\opdelta^{\left[k\right]}_{\opd^{p}l^{\left(p\right)}}
= \frac{1}{k^{\abs{\widetilde{C}_{q}}}}
\sum\limits_{\widetilde{m}^{\left(q\right)}\in\widetilde{C}^{q}_{k}}
e^{\frac{2\pi i}{k}\widetilde{m}^{\left(q\right)}\cdot\opd^{p} l^{\left(p\right)}}
\end{align*}
where $n = p + q + 1$, and where $l^{\left(p\right)}$ is regarded as a vector with $\abs{C_{p}}$ components, $\opd^{p}$ as a matrix of size $\abs{C_{p}}\times\abs{C_{p+1}}$ and $\widetilde{m}^{\left(q\right)}$ as a vector of size $\abs{\widetilde{C}_{q}} = \abs{C_{n-q}}=\abs{C_{p+1}}$ by construction.
\end{claim}

\begin{consequence} 
We can write
\begin{align}
\nonumber
&\mathcal{Z}^{p}_{\mathrm{TV}_{k}}\left(M^{\left(n\right)}\right)\\
&\hspace{0.5cm}= \frac{1}{k^{\abs{C_{p+1}}+\abs{C_{p-1}}-\abs{C_{p-2}}+\hdots\pm\varepsilon_{p}\abs{C_{0}}\mp\varepsilon_{p}}}
\sum\limits_{l^{\left(p\right)}_{k}\in C^{p}_{k}}\sum\limits_{\widetilde{m}^{\left(q\right)}_{k}\in\widetilde{C}^{q}_{k}}
e^{\frac{2\pi i}{k}\widetilde{m}^{\left(q\right)}\cdot\opd^{p} l^{\left(p\right)}}
\end{align}
which can be interpreted as a sort of partition function of a discrete BF theory.
\end{consequence}

\section{Conclusion}

\begin{proposition}[Ph. M., E. H{\o}ssjer, F. Thuillier]
\begin{align}
\mathcal{Z}^{p}_{\mathrm{BF}_{k}}\left(M\right)
= \frac{\abs{T_{p}}}{\abs{\Hom{F_{p}}{\slfrac{\Z}{k\Z}}}}
\frac{\abs{\Hom{F_{p-1}}{\slfrac{\Z}{k\Z}}}\hdots}{\abs{\Hom{F_{p-2}}{\slfrac{\Z}{k\Z}}}\hdots}
\mathcal{Z}^{p}_{\mathrm{TV}_{k}}\left(M\right)
\end{align}
\end{proposition}

\begin{remark}
We recover the case $n = 3$ and $p = 1$: 
\begin{equation}
\label{BF=TV}
\mathcal{Z}^{1}_{\mathrm{BF}_{k}}\left(M^{\left(3\right)}\right)
= \frac{\abs{T_{1}}}{\abs{\Hom{F_{1}}{\slfrac{\Z}{k\Z}}}}
\mathcal{Z}^{1}_{\mathrm{TV}_{k}}\left(M^{\left(3\right)}\right)
\end{equation}
\end{remark}

\begin{remark}
The appearance of the $\Hom{F_{p}}{\slfrac{\Z}{k\Z}}$ may remind us of the ghost, antifield and antighost sectors that appear in Batalin-Vilkovisky formalism \cite{MNE2019}. An analysis of the normalizing part \eqref{Normalization} through this formalism may result in some factors that would cancel the $\Hom{F_{p}}{\slfrac{\Z}{k\Z}}$ factors in \eqref{BF=TV}, leading to a better agreement between the Abelian BF and TV theories.
\end{remark}

Let us try to give a quick overview of the significance of this result. In the $\mathrm{SU}\!\left(2\right)$ 3D case, the TV construction relying on $\mathcal{U}_{q}\!\left(\mathfrak{sl}_{2}\!\left(\mathbb{C}\right)\right)$ is regarded as a regularization of the $\mathrm{SU}\!\left(2\right)$ BF partition function\footnote{It is often said they are equal, but the point is that the BF partition function is not a well-defined object, so the equality does not really make sense, \textit{stricto sensu}.}. 

In the $\mathrm{U}\!\left(1\right)$ 3D case, we are able to propose a non-perturbative definition of the partition function on the one hand, and we are able to build an Abelian realization of a spherical category depending on a ``cutting parameter'' $k$, and to use it to build an Abelian TV invariant on the other hand. The two quantities obtained independently turn out to be equal up to a factor that has a cohomological interpretation. 

We are able to generalize these Abelian BF and TV constructions to a family of constructions in any dimension. Those constructions are still related, again up to a factor that has a cohomological interpretation.

Remark that formula \eqref{BF=TV} is an equality between sums of exponentials of two different bilinear forms over two different lattices. In \cite{DT2007}, Florian Deloup and Vladimir Turaev showed a similar equality in a purely algebraic manner, and they also ended up with a factor between the two sums, that can be interpreted as a normalization by the volume of each lattice.

% \textcolor{red}{Add the construction of DB cohomology, truncated double complex etc... in appendix?}

\subsection*{Acknowledgments} 

Ph. M. thanks Pavel Mn{\"e}v, Stephan Stolz, Liviu Nicolaescu, Christopher Schommer-Pries, Justin Beck and Eugene Rabinovich from the Department of Mathematics of the University of Notre Dame, Vivek Shende and Konstantin Wernli from the Center for Quantum Mathematics of Southern Denmark University, S{\'e}bastien Leurent from the Institut de Math{\'e}matiques de Bourgogne of the Universit{\'e} de Bourgogne, Dominique Manchon and Abel Lacabanne from the Laboratoire Blaise Pascal of the Universit{\'e} Clermont-Auvergne, for their interest and the useful remarks and suggestions they made in order to improve the presentations of this work. Finally, Ph. M. thanks Ralph Kaufmann from the Department of Mathematics of Purdue University for his invitation to the AMS Special Session on Higher Structures in Topology, Geometry and Physics in March 2022.

\bibliographystyle{amsalpha}

\end{document}